\documentclass{sig-alternate-05-2015}

\usepackage{subfigure}
\usepackage[usenames]{color}
\usepackage{epsfig}
\usepackage{graphicx}
\usepackage{amsmath}
\usepackage[ruled]{algorithm2e}

\newdef{Claim}{Claim}
\newdef{Example}{Example}

\newtheorem{prop}{Proposition}
\newtheorem{thm}{Theorem}

\begin{document}

\title{Costly Circuits, Submodular Schedules:\\ Hybrid Switch Scheduling for Data Centers}

\numberofauthors{3} %  in this sample file, there are a *total*

\author{
% 1st. author
\alignauthor
Shaileshh Bojja Venkatakrishnan\\
%       \affaddr{Department of Electrical and Computer Engineering}\\
       \affaddr{University of Illinois Urbana-Champaign}\\
       \email{bjjvnkt2@illinois.edu}
% 2nd. author
\alignauthor
Mohammad Alizadeh\\
%       \affaddr{Department of Electrical Engineering and Computer Science}\\
       \affaddr{Massachusetts Institute of Technology}\\
       \email{alizadeh@csail.mit.edu}
\alignauthor
Pramod Viswanath\\
%       \affaddr{Department of Electrical and Computer Engineering}\\
       \affaddr{University of Illinois Urbana-Champaign}\\
       \email{pramodv@illinois.edu}
}

% There's nothing stopping you putting the seventh, eighth, etc.
% author on the opening page (as the 'third row') but we ask,
% for aesthetic reasons that you place these 'additional authors'
% in the \additional authors block, viz.

% Just remember to make sure that the TOTAL number of authors
% is the number that will appear on the first page PLUS the
% number that will appear in the \additionalauthors section.

\maketitle
\begin{abstract}
Hybrid switching -- in which a high bandwidth circuit switch (optical or wireless) is used in conjunction with a low bandwidth packet switch -- is a promising alternative to interconnect servers in today's large scale data-centers. Circuit switches offer a very high link rate, but incur a non-trivial reconfiguration delay which makes their scheduling challenging. In this paper, we demonstrate a lightweight, simple and nearly-optimal scheduling algorithm that trades-off configuration costs with the benefits of reconfiguration that match the traffic demands. The algorithm has strong connections to submodular optimization, has performance at least half that of the optimal schedule and strictly outperforms state of the art in a variety of traffic demand settings. These ideas naturally generalize: we see that indirect routing leads to exponential connectivity; this is another phenomenon of the power of multi hop routing, distinct from the well-known load balancing effects.
\end{abstract}

%A category including the fourth, optional field follows...
%\category{D.2.8}{Software Engineering}{Metrics}[complexity measures, performance measures]


%\keywords{ACM proceedings, \LaTeX, text tagging} % NOT required for Proceedings

\section{Introduction}

Modern data centers are massively scaling up to support demanding
applications such as large-scale web services, big data analytics, and
cloud computing. The computation in these applications is distributed
across tens of thousands of interconnected servers. As the number and
speed of servers increases,\footnote{ Servers with 10Gbps network
  interfaces are common today and 40/100Gbps servers are being
  deployed.} providing a fast, dynamic, and economic switching
internconnect in data centers constitutes a topical networking
challenging. Typically, data center networks use multi-rooted tree
designs: the servers are arranged in racks and an Ethernet switch at
top of the rack (ToR) connects the rack of servers to a one or more
aggregation (or spine) layers. These designs use multiple paths
between the ToRs to deliver uniform high bisection bandwidth, and
consist of a large number of high speed electronic packet switches
that provide fine-grained switching capabilities but at poor
speed/cost ratios.

Recent work has proposed the use of high speed circuit switches based
on optical~\cite{180616,
  farrington2011helios, wang2011c} or
wireless~\cite{kandula2009flyways, zhou2012mirror,
  hamedazimi2014firefly} links to interconnect the ToRs. These
architecures enable a dynamic topology tuned to actual traffic
patterns, and can provide a much higher aggregate capacity than a
network of electronic switches at the same price point, consume
significantly less power, and reduce cabling complexity. For instance,
Farrington~\cite{farrington2012optics} reports 2.8$\times$, 6$\times$,
4.7$\times$ lower cost, power, and cabling complexity using optical
circuit switching relative to a baseline network of electronic
switches.
% a rule of thumb that has stayed steady over the years has been that
% optical switches provide roughy 10 times the throughput of electronic
% ones at the same price. 

The drawback of circuit switches, however, is that their switching
configuration time is much slower than electronic switches. Depending
on the specific technology, reconfiguring the circuit switch can take
a few milliseconds (e.g., for 3D MEMS optical circuit
switches~\cite{180616,
  farrington2011helios, wang2011c}) to 10s of microseconds (e.g., for
2D MEMs wavelength-selective switches~\cite{porter2013integrating}).
During this reconfiguration period, the circuit switch cannot carry
any traffic. By contrast, electronic switches can make per-packet
switching decisions at sub-microsecond timescales. This makes the
circuit switch suitable for routing stable traffic or bursts of
packets (e.g., hundreds to thousands of packets at a time), but not
for sporadic traffic or latency sensitive packets. A natural approach
is then to have a {\em hybrid} circuit/packet switch architecture: the
circuit switch can handle traffic flows that have heavy intensity but
also require sparse connections, while a lower capacity packet switch
handles the complementary (low intensity, but dense connections)
traffic flows~\cite{farrington2011helios}.

% In practice, the circuit switch is typically an optical
% switch~\cite{180616, farrington2011helios,
%   wang2011c}.\footnote{Designs based on point-to-point wireless links
%   have also been proposed~\cite{kandula2009flyways,
%     zhou2012mirror}. Our abstract model is general.}  These switches
% have a key limitation: changing the circuit configuration imposes a
% {\em reconfiguration delay} during which the switch cannot carry any
% traffic. The reconfiguration delay can range from few milliseconds to
% 10s of microseconds depending on the
% technology~\cite{porter2013integrating, liu2014circuit}. This makes
% the circuit switch suitable for routing stable traffic or bursts of
% packets (e.g., hundreds to thousands of packets at a time), but not
% for sporadic traffic or latency sensitive packets. Therefore, hybrid
% networks also use a (electrical) packet switch to carry traffic that
% cannot be handled by the circuit switch. The packet switch operates on
% a packet-by-packet basis, but has a much lower capacity than the
% circuit switch. For example, the circuit and packet switches might
% respectively run at 100Gbps and 10Gbps per port.

With this hybrid architecture, the relatively low intensity traffic is
taken care of by the packet switch --- switch scheduling here can be
done dynamically based on the traffic arrival and is a well studied
topic~\cite{mckeown1999achieving, keslassy2005optimal,
  mckeown1999islip}. On the other hand, scheduling the circuit switch,
based on the heavy traffic demand matrix, is still a fundamental
unresolved question. Consider an architecture where a centralized
scheduler samples the traffic requirements at each of the ToR ports at
regular intervals ($W$, of the order of 100$\mu$s--1ms), and looks to
find the schedule of circuit switch configurations over the interval
of $W$ that is ``matched" to the traffic requirements. The challenge
is to balance the overhead of reconfiguration the circuits
% ($\delta$ time units, typically of the order of $10-100\mu$s) 
with the capability to be flexible and meet the traffic demand
configuration.

The centralized scheduler must essentially decide a sequence of {\em
  matchings} between sending and receiving ToRs which the circuit
switch then implements. For an optical circuit switch, for instance,
the switch realizes the schedule by appropriately configuring its MEMs
mirrors. As another example, in a broadcast-select optical ring
architecture~\cite{cao2014joint}, the ToRs implement the controller's
schedule by tuning in to the appropriate wavelength to receive traffic
from their matching sender as dictated by the schedule.

% Discuss a bit on how optical switch architectures look like. Suggest
% optical ring network -- allows for all ToR to be connected in a ring
% and orthogonal frequencies/mirrors are used for tx/rx of data. Any
% choice of tx/rx frequencies gives a {\em matching} -- a permutation
% matrix. Discuss current technological trends w.r.t. optical link
% speeds and switching-delays.

Hence, we need a scheduling algorithm that decides the state (i.e.,
matching) of the circuit switch at each time and also a routing
protocol to decide on an appropriate route packets can take to reach
their destination ToR port. This is a challenging problem and entails
making several choices on: (a) number of matchings, (b) choice of
matchings (switch configuration), (c) durations of the matchings and
(d) the routing protocol, in each interval $W$. Mathematically, this
leads to a well defined optimization problem, albeit involving both
combinatorial and real-valued variables. Even special cases of this
problem~\cite{li2003scheduling} are NP hard to solve exactly. Recent
papers have proposed heuristic algorithms to address this scheduling
problem. In~\cite{Solstice} the authors present Solstice --- a greedy
perfect-matching based heuristic for a hybrid electrical-optical
switch. Experimental evaluations show Solstice performing well over a
simple baseline (where the schedules are provided by a truncated
Birkoff-von Neumann decomposition of the traffic matrix), although no
theoretical guarantees are presented. Indirect routing in a
distributed setting, but without considerations of configuration
switching costs, is studied in another recent work
\cite{cao2014joint}.

\subsection{Our Contributions}
We first focus on routing policies where packets are sent from the
source port to the destination port only via a direct link connecting
the two ports, leading to {\em direct} or {\em single-hop} routing.
Our main result here is an approximately optimal, very simple and fast
algorithm for computing the switch schedule in each interval. The
algorithm, which we christen Eclipse, has a performance that is at
least half that of optimal for every instance of the traffic demands,
and experimentally shows a strict and consistent improvement over the
state-of-the-art. A key technical contribution here is the
identification of a submodularity structure~\cite{azar2012efficient}
in the problem, which allows us to make connections to submodular
function maximization and the circuit switch scheduling problem with
reconfiguration delay.

Next, we consider routing polices where packets are allowed to reach
their destination after (potentially) transiting through many
intermediate ports, leading to {\em indirect} or {\em multi-hop}
routing. This class of routing policies is motivated by our
observation that if the number of matchings is limited, multi-hop
routing can {\em exponentially} improve the {\em reachability} of
nodes; a novel benefit of multi-hop routing distinct from the
classical and well known load balancing effects
\cite{rabin1989efficient,valiant1990bridging,greenberg2008towards}. We
again identify submodularity in the problem, but the constraints for
this submodular maximization problem are no longer linear and
efficient solutions challenging to find. However, for the important
special case where the sequence of switch configurations have already
been calculated (and the indirect routing policy has to be decided) we
propose a simple and fast greedy algorithm that is near optimal
universally for all traffic requirements. Detailed experimental
demonstrate strong improvements over direct routing, which are
especially pronounced when the switch reconfiguration delays are
relatively large.

The paper is organized as follows. In Section~\ref{sec: System Model}
the model, framework and the problem objective are formally stated
along with a succinct summary of the state of the art.
Section~\ref{sec:direct} focuses on direct routing and
Section~\ref{sec: Indirect Routing} on indirect routing. In
Section~\ref{sec:Evaluations} we present a detailed evaluation of the
proposed algorithms on a variety of traffic inputs. Section~\ref{sec:conclusion} closes with a brief discussion. Technical aspects of the
algorithm and its evaluation, including connections to submodularity
and combinatorial optimization problems are deferred to the
supplementary material.

\section{System Model} \label{sec: System Model} 

In this section, we present our model for a hybrid
circuit/packet-switched network fabric, and formally define our
scheduling problem. Our model closely follows~\cite{Solstice}.

%Following~\cite{Solstice} we consider a hybrid electronic-optical switch fabric. 

\subsection{Hybrid Switch Model}
We consider an $n$-port network where each port is simultaneously
connected to a circuit switch and a packet switch. A set of nodes are
attached to the ports and communicate over the network. The nodes
could either be individual servers or ``top-of-rack'' switches.

We model the circuit switch as an $n\times n$ cross-bar comprising of
$n$ input ports and $n$ output ports. At any point in time, each input
port can send packets to at most one output port and each output port
can receive packets from at most one input port over the circuit
switch. The circuit switch can be reconfigured to change the
input-output connections. We assume that the packets at the input
ports are organized in
virtual-output-queues~\cite{prabhakar1999speedup} (VOQ) which hold
packets destined to different output ports.

In practice, the circuit switch is typically an optical
switch~\cite{180616, farrington2011helios,
  wang2011c}.\footnote{Designs based on point-to-point wireless links
  have also been proposed~\cite{kandula2009flyways,
    zhou2012mirror}. Our abstract model is general.}  These switches
have a key limitation: changing the circuit configuration imposes a
{\em reconfiguration delay} during which the switch cannot carry any
traffic. The reconfiguration delay can range from few milliseconds to
10s of microseconds depending on the
technology~\cite{porter2013integrating, liu2014circuit}. This makes
the circuit switch suitable for routing stable traffic or bursts of
packets (e.g., hundreds to thousands of packets at a time), but not
for sporadic traffic or latency sensitive packets. Therefore, hybrid
networks also use a (electrical) packet switch to carry traffic that
cannot be handled by the circuit switch. The packet switch operates on
a packet-by-packet basis, but has a much lower capacity than the
circuit switch. For example, the circuit and packet switches might
respectively run at 100Gbps and 10Gbps per port.

We divide time into slots, with each slot corresponding to a
(full-sized) packet transmission time on the circuit switch.  We
consider a {\em scheduling window} of $W\in\mathbb{Z}$ time units.  A
central controller uses measurements of the aggregated traffic demand
between different ports to determine a {\em schedule} for the circuit
switch at the start of each scheduling window. The schedule comprises
of a sequence of configurations and how long to use each configuration
($\S$\ref{sec:objective}). We assume that the delay for each
reconfiguration is $\delta\in\mathbb{Z}$ time units.

% A central scheduler estimates the traffic demand at the start of each
% scheduling window, and uses this information to determine a schedule
% (i.e., a sequence of circuit configurations) for the following $W$
% time units. Our work is orthogonal to how the scheduler obtains the
% traffic demand estimate. For example, it could be as simple as
% aggregating

%\smallskip
%\noindent{\bf Schedule.}

% Note that the state of the cross-bar switch at any instant can be
% described by a $n \times n$ permutation matrix $P$, where for input
% port $i$ we have $P(i,j) = $ if it is connected to the $j$th output
% port. As such, we describe the sequence of switch configurations by
% $(\alpha_1, P_1), (\alpha_2, P_2),\ldots, (\alpha_k, P_k)$ where
% $\alpha_i\in\mathbb{Z}$ denotes the duration of the $i$th switch
% configuration $P_i$. 

%\smallskip
%\noindent{\bf Traffic Demand.}

\subsection{Traffic Demand}
Let $T\in\mathbb{Z}^{n\times n}$ denote the accumulated traffic at the
start of a scheduling window. We assume $T$ is a feasible traffic
demand, i.e., $T$ is such that $\sum_{j=1}^n T(i,j) \leq W$ and
$\sum_{i=1}^n T(i,j) \leq W$ for all $i,j \in \{1,2,\ldots,n\}$. The
$(i,j)$th entry of $T$ denotes the amount of traffic that is in the
VOQ at node $i$ destined for node $j$. 

We assume that the controller knows $T$.\footnote{Our work is
  orthogonal to how the controller obtains the traffic demand
  estimate.  For example, the nodes could simply report their backlogs
  before each scheduling window, or a more sophisticated prediction
  algorithm could be used.} We also assume that non-zero entries in
the traffic matrix $T$ are bounded as
$2\delta \leq T(i,j) \leq \epsilon W$ for all $i,j\in[n]:T(i,j)>0$ and
some parameter $0<\epsilon < 1$. This is a mild condition because
traffic between pairs of ports that is small relative to $\delta$ is
better served by the packet switch anyway.

%% 
%% MA: The upper bound on T is not easy to justify. Is this only a
%% technical condition needed to make the proof to go through, or is
%% it actually necessary for the algorithm?
%%

Previous measurement studies have shown that the inter-rack traffic in
production data centers is sparse~\cite{Solstice, benson2010network,
  alizadeh2011data, roy2015inside}. Over short periods of time (e.g.,
10s of milliseconds), most nodes communicate with only a small number
of other nodes (e.g., few to low tens). Further, in many cases, a
large fraction of the traffic is sent by a small fraction of
``elephant'' flows~\cite{alizadeh2011data}. While our algorithms and
analysis are general, it is important to note that such sparse traffic
patterns are necessary for hybrid networks to perform well (especially
with larger reconfiguration delay).

%  tried to characterize inter-server
% traffic flow in data centers. While early studies were conducted via
% simulations or test-beds, a number of subsequent studies have been
% conducted on production data-centers of Microsoft and Facebook.

%\smallskip
%\noindent{\bf 
\subsection{The Scheduling Problem}
\label{sec:objective}

Given the traffic demand, $T$, our goal is to compute a schedule that
maximizes the total amount of traffic sent over the circuit switch
during the scheduling window $W$. This is desirable to minimize the
load on the slower packet switch. In general, the scheduling problem
involves two aspects:

%(1) determining the circuit switch configuration
%schedule; (2) deciding how to route traffic over the circuit switch.

%\smallskip
\noindent{\bf 1. Determining a schedule of circuit switch
  configurations:}
The algorithm must determine a sequence of circuit switch
configurations:
$(\alpha_1, P_1), (\alpha_2, P_2),\ldots, (\alpha_k, P_k).$
%\begin{displaymath}
%(\alpha_1, P_1), (\alpha_2, P_2),\ldots, (\alpha_k, P_k).
%\end{displaymath} 
Here, $\alpha_i\in\mathbb{Z}$ denotes the duration of the $i^{th}$
switch configuration, and $P_i$ is an $n \times n$ permutation matrix,
where $P_i(s,t)=1$ if input port $s$ is connected to output port $t$
in the $i^{th}$ configuration. For a valid schedule, we must have
$\alpha_1 + \alpha_2 + \ldots + \alpha_k + k\delta \leq W$ since the
total duration of the configurations cannot exceed the scheduling
window $W$.

%\smallskip
\noindent{\bf 2. Deciding how to route traffic:}
The simplest approach is to use only direct routes over the circuit
switch. In other words, each node only sends traffic to destinations
to which it has a direct circuit during the scheduling
window. Alternatively, we can allow nodes to use indirect routes,
where some traffic is forwarded via (potentially multiple)
intermediate nodes before being delivered to the destination. Here,
the intermediate nodes buffer traffic in their VOQs for transmission
over a circuit in a subsequent configuration.

%As we discuss in
%detail in $\S$\ref{sec:indirect}, indirect routing can be advantageous
%because it allows nodes that have no direct connection over the
%circuit switch to communicate. 

%that otherwise would not have
%an p

In the next section, we begin by formally defining the problem in the
simpler setting with direct routing and developing an algorithm for
this case. Then, in $\S$\ref{sec: Indirect Routing}, we consider the more
general setting with indirect routing.

\noindent {\bf Remark 1.} 
Prior work~\cite{Solstice, li2003scheduling} has considered the
objective of {\em covering} the entire traffic demand in the least
amount of time. For example, the \textsc{ADJUST} algorithm
in~\cite{li2003scheduling} takes the traffic demand $T$ as input and
computes a schedule $(\alpha_1, P_1),\ldots,(\alpha_k, P_k)$ such that
$\sum_{i=1}^k \alpha_i P_i \geq T$ and
$\sum_{i=1}^k \alpha_i + k\delta$ is minimized. Our formulation (and
solution) is more general, since an algorithm which maximizes
throughput over a given time period can also be used to find the
shortest duration to cover the traffic demand (e.g., via binary
search).

% provides a better theoretical understanding of the underlying
% problem structure. This in turn can guide us better design situation
% specific algorithms in practice.

% \subsection{Notation}
% traffic graph; round of switching; $\mathbf{J}_n$ - $n$ dimensional
% circularly shifted identity matrix.

\subsection{Related Work}
Before presenting the work in this paper, we briefly summarize related work on this topic. 
Scheduling in crossbar switches is a classical and well studied topic and traditionally it has been used to model the packet switch where the reconfiguration delay is very small. Hence the scheduling solutions proposed -- ranging from centralized Birkhoff-von-Neumann decomposition scheduler~\cite{mckeown1999achieving} on one end to the decentralized load-balanced scheduler~\cite{chang2001load} on the other -- did not account for reconfiguration delay. In a different context (satellite-switched time-division multiple access), works such as~\cite{gopal1985minimizing} computed  schedules that minimized the number of matchings in the schedule.  

With the proposals on hybrid circuit/packet switching systems~\cite{farrington2011helios,wang2011c}, simplified models that factor for the reconfiguration delay were considered. Early works often assumed the delay to be either zero~\cite{inukai1979efficient} or infinity~\cite{ towles2003guaranteed, wu2006nxg05}. The infinite delay setting corresponds to a problem where the number of matchings is minimized. However they still require $O(n)$ matchings. Moderate reconfiguration delays are considered in DOUBLE~\cite{towles2003guaranteed} and other algorithms such as~ \cite{fu2013cost, li2003scheduling, wu2006nxg06} that explicitly take reconfiguration delay into account. The algorithm ADJUST~\cite{li2003scheduling} minmizes the covering time but still requires around $n$ configurations. All of these  algorithms do not benefit from sparse demands and continue to require $O(n)$ configurations~\cite{Solstice}. In a complementary approach,~\cite{dasylva1999optimal} considers conditions on the input traffic matrix under which efficient polynomial time algorithms to compute the optimal schedule exists. Yet other approaches have been to introduce speedup~\cite{mckeown1999islip}, or randomization in the algorithms~\cite{giaccone2003randomized}, however they dont address the basic optimization problem underlying this scenario head-on. Such is the goal of this paper.  

The above algorithms are ``batch'' policies~\cite{wang2015adaptive} in which each computational call returns a schedule for an entire window of time. Another research direction is to consider ``dynamic'' policies where scheduling decisions are made time-slot by time-slot. A variant of the well known MaxWeight algorithm is presented in~\cite{wang2015end} and is shown to be throughput optimal. Fixed-Frame MaxWeight (FFMW) is a frame based policy proposed in~\cite{li2003frame} and has good delay performance. However it requires the arrival statistics to be known in advance. A hysteresis based algorithm that adapts many previously proposed algorithms for crossbar switch scheduling is presented in~\cite{wang2015adaptive}. All these algorithms require perfect queue state information at every instant.

\section{Direct Routing}
\label{sec:direct}
The centralized scheduler samples the ToR ports and arrives at the traffic demand (matrix) $T$ to be met in the upcoming slot. 
In this section, we develop an algorithm, named Eclipse, that takes the traffic demand, $T$,
as input and computes a schedule of matchings (circuit configurations)
and their durations to maximize throughput over the circuit switch; only direct routing of packets from source to destination ports are allowed here. Eclipse is fast, simple and nearly-optimal in every instance of the traffic matrix $T$. Towards a formal understanding of the notion of optimality, 
 consider the following optimization problem:
\begin{align}
\mathrm{maximize} \quad  \left\| \min \left( \sum_{i=1}^k \alpha_i P_i, T \right) \right\|_1  \label{eq:opt problem} \\
\mathrm{s.t.} \quad \alpha_1 + \alpha_2 + \ldots + \alpha_k + k\delta \leq W \\
k \in \mathbb{N}, P_i \in \mathcal{P}, \alpha_i \geq 0~ \forall i\in\{1,2,\ldots,k\}
\end{align}
where $\mathbb{N} = \{1,2,\ldots\}$ and $\mathcal{P}$ is the set
of permutation matrices.

This optimization problem is NP-hard~\cite{li2003scheduling}, and recent work \cite{Solstice} in the literature has focused on heuristic solutions.  Our proposed algorithm has some similarities 
to the prior work in~\cite{Solstice} in that the matchings and their
durations are computed successively in a {\em greedy} fashion. However, the algorithm is overall quite different in terms of both details and ideas; we 
uncover and exploit the underlying {\em submodularity}~\cite{schrijver-book}
structure inherent in the problem to design and analyze the algorithm in a
{\em principled} way.

%and prove its approximation guarantees.

\subsection{Intuition}
\label{sec: Motivation}
Before a formal presentation and analysis of the algorithm, we begin with an intuitive and less-formal approach to how one might solve this optimization problem.  Consider greedy algorithms with the template shown in
Algorithm~\ref{alg: heuristic}. The template starts with an empty
schedule, and proceeds to add a new matching to the schedule in each
iteration. This process continues until the total duration of the matchings
exceeds the allotted time budget of $W$, at which point the algorithm
terminates and outputs the schedule computed so far. In each
iteration, the algorithm first picks the duration of the matching,
$\alpha$. It then selects the {\em maximum weight matching} in the
traffic graph whose edge weights are thresholded by $\alpha$ (i.e.,
edge weights $>\alpha$ are clipped to $\alpha$). The traffic graph is
a bipartite graph between $n$ input and $n$ output vertices, with an
edge of weight $T(i,j)$ between input node $i$ and output node $j$. It
only remains to specify how to choose $\alpha$ in each iteration.

\begin{algorithm}[!tbp]
    \SetKwInOut{Input}{Input}
    \SetKwInOut{Output}{Output}
    \Input{Traffic demand $T$, reconfiguration delay $\delta$ and scheduling window size $W$}
    \Output{Sequence of matchings and weights: $(\alpha_1,P_1),\ldots,(\alpha_k, P_k)$}
$\mathrm{sch} \leftarrow \{ \}$  \tcp*{schedule}
$k \leftarrow 0$ \;
$T_\mathrm{rem} \leftarrow T$ \tcp*{traffic remaining}
\While{$\sum_{i=1}^k (\alpha_i + \delta) \leq W$}
{
$k \leftarrow k + 1$\;
\textit{Decide on a duration $\alpha$ for the matching}\;
$M \leftarrow \mathrm{argmax}_{M\in\mathcal{M}}  \|\min(\alpha M, T_\mathrm{rem})\|_1  $ \;
$\mathrm{sch} \leftarrow \mathrm{sch} \cup \{(\alpha, M)\}$ \;
$T_\mathrm{rem} \leftarrow T_\mathrm{rem} - \min(\alpha M, T_\mathrm{rem})$ \;
}
\If{$\sum_{i=1}^k (\alpha_i + \delta) > W$}{$\mathrm{sch} \leftarrow \mathrm{sch} \backslash \{(\alpha, M)\}$\; }
    \caption{A general greedy algorithm template} \label{alg: heuristic}
\end{algorithm}

Consider an exercise where we vary the matching duration $\alpha$ from
$0$ to $W$ and compute the maximum weight matching in the thresholded
traffic graph for each $\alpha$. For a typical traffic matrix, this
results in a curve similar to the solid-blue line in
Fig.~\ref{fig:utilization}. Notice that the value of the maximum
weight matching is precisely equal to the sum-throughput that can be
achieved in that round of the switch schedule. It is straightforward to see that the maximum weight matching curve has the following properties: (a) it
is non-decreasing and (b) piecewise linear. These are explained as
follows: when $\alpha$ is very small a lot of the edges in the traffic
graph have a weight that is saturated at $\alpha$. Hence it is likely
to find a perfect matching with total weight of $n\alpha$. As such the
slope of the curve when $\alpha$ is small is $n$. However, as $\alpha$
becomes large there are increasingly fewer edges whose
weights are saturated at $\alpha$ and, correspondingly, theslope reduces. When
$\alpha$ is so large that all of the edge weights are strictly smaller
than $\alpha$, then the value of the maximum weight matching does not
change even with any further increase in $\alpha$ and the curve
ultimately flattens out.

%\XXXNote{MA: Do we need the text starting with ``This can be explained as...''}

\begin{figure}[!tbp]
  \centerline{\includegraphics[height=70mm]{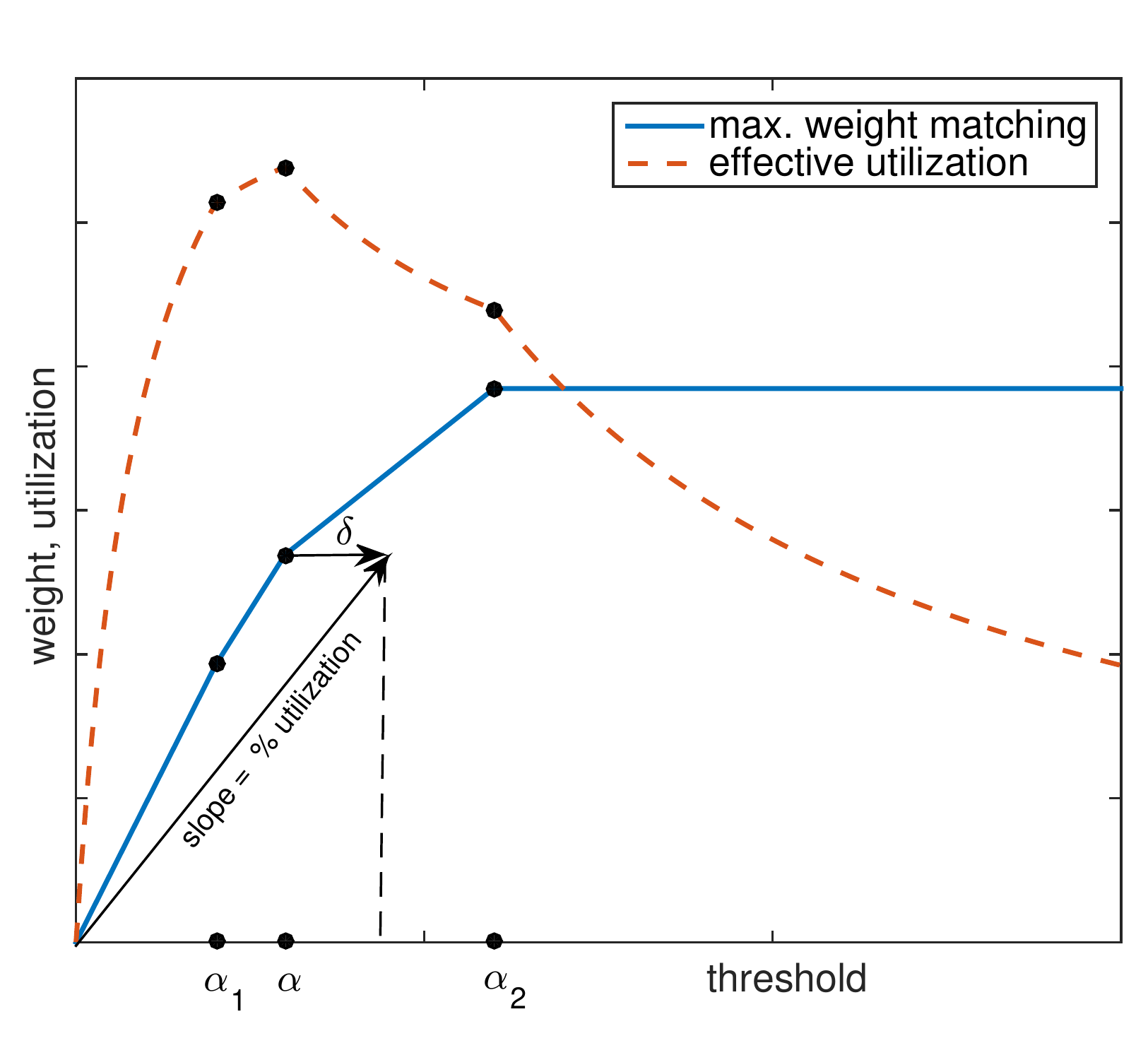}}
    \caption{Throughput of max. weight matching as a function of threshold duration. The effective utilization curve of the matchings is also shown.}
    \label{fig:utilization}
\end{figure}

Two operating points of interest, considering
Fig.~\ref{fig:utilization}, are (a) the largest $\alpha$ where the
slope of the curve is maximum ($=n$ in the typical case where every ingress/egress port has traffic)
and (b) the smallest $\alpha$ where the value of the maximum weight
matching is the largest. These points have been denoted by $\alpha_1$
and $\alpha_2$ in Fig.~\ref{fig:utilization} respectively. Setting
$\alpha=\alpha_1$ is interesting because it results in a matching
where the links are all fully utilized. For example, the
Solstice algorithm presented in~\cite{Solstice} implicitly
adopts this operating point. On the other hand, $\alpha=\alpha_2$
gives a matching that achieves the largest possible sum-throughput in
that round.

However we note that both choices of $\alpha$ are less than ideal for
the following reasons.  Recall that after every round of switching we
incur  a delay of $\delta$ time units. As such if the value of
$\alpha_1$ is  small (say $\O(\delta)$) in each round, then the number
of matchings, and hence the time wasted due to the reconfiguration
delay, becomes large. As a concrete example, consider the transpose of the traffic
matrix $T_1 = [A_1^t {\bf b}_1^t]$ where  $A_1$ is a sparse $(n-1)\times n$ matrix and 
$\mathbf{b} = [2\delta, 2\delta, \ldots, 2\delta, 0,0,\ldots,0]$ 
comprises of some $k$ entries of value $2\delta$ and $n-k$ entries of
value 0. In other words, we are considering an input where a node or a
collection of nodes have a large number of small flows to a particular
node or vice-versa. For such an instance it is clear that if we insist
on matchings with 100\% utilized links, then the maximum duration of
the matching is $2\delta$ (i.e., $\alpha_1 = 2\delta)$. Thus,
continuing the process described in Algorithm~\ref{alg: heuristic}
results in a sequence of $k$ matchings each of which is only $2\delta$
time units long. Hence in the worst case (if $k > 1/(3\delta)$) about
$1/3$rd of the entire scheduling window is wasted just due to
reconfiguration delay limiting the maximum possible throughput to
$2n/3$. On the other hand, if we had ignored the entries in
$\mathbf{b}$, then we could have scheduled just $A_1$ achieving a
total throughput of $n - 2k\delta \approx n$ for large $n$. We
point out that the phenomenon described above happens in a large family of instances, of which $T_1$ is a specific example. We also emphasize that such instances are  pretty likely to occur in
practice; for example,~\cite[Fig.5-b]{roy2015inside} shows traffic
measurements in a Facebook datacenter where the interactions between
Cache and Webservers lead to traffic matrices having this property.

Similarly for the operating point with $\alpha = \alpha_2$,  consider the traffic matrix $T_2 = \left[ \begin{array}{cc} A_2 & \mathbf{0} \\   \mathbf{0} & B_2 \end{array} \right]$ where $A_2$ is a sparse $(n-2)\times(n-2)$ matrix and
%\begin{align}
%A_2 &= \left[ \begin{array}{ccccccc}
%0 & \frac{1}{3} & \frac{1}{3} & \frac{1}{3} & 0 & 0 & \cdots \\
%0 & 0 & \frac{1}{3} & \frac{1}{3} & \frac{1}{3} & 0 &\cdots \\
%\vdots & \vdots & \vdots & \vdots & \vdots & \ddots & \vdots \\
%\frac{1}{3} & \frac{1}{3} & \frac{1}{3} & 0 & 0 & \cdots & 0
%\end{array} \right]_{n-2\times n-2},
%\end{align}
$
B_2 = \left[\begin{array}{cc} 0 & 1 \\ 1 & 0 \end{array}\right].
$
This is a diametrically opposite situation from $T_1$ where a small collection of nodes interact only amongst themselves with no interaction outside. Such a situation occurs, for example, in multi-tenant cloud-computing datacenters~\cite{shieh2010seawall} where individual tenants run their jobs on small clusters of servers. In such a case, the value of the maximum weight matching can be maximum for a large $\alpha$. For $T_2$ the maximum value occurs at $\alpha = 1-\delta$ (i.e., $\alpha_2 = 1-\delta$), resulting in a schedule with just one matching of duration $1-\delta$ and potentially missing a lot of traffic for $A_2$. For example, if $A_2$ is uniformly $k$-sparse, we miss out roughly $(k-1)n/k$ units of traffic. On the other hand, by choosing the duration of the matching to be $1/k - \delta$ in each step we can achieve a sum throughput of $n - O(\delta) \approx n$.

In scenarios exemplified by $T_2$, setting $\alpha = 1$ is bad because the utilization of the resulting matching is poor, i.e., a vast majority of the matching links carry only a fraction of their capacity. This can be overcome by insisting that we choose only those matchings with utilization of at least 75\% (say). However, in the case of $T_1$ we observe a poor performance in spite of all matchings having a utilization of 100\%. The issue in this case is that the duration of the matchings are small compared to the reconfiguration delay. Hence to avoid this scenario we can insist on $\alpha \geq 20\delta$ (say) in Algorithm~\ref{alg: heuristic}. Our first main observation is that both of the above heuristics are captured if we consider the {\em effective utilization} of the matchings. We define effective utilization as the ratio mwm$(\alpha)/(\alpha + \delta)$ where mwm$(\alpha)$ denotes the value of the maximum weight matching at $\alpha$. This ratio indicates the overall efficiency of a matching by  including the reconfiguration delay into the duration. In Fig.~\ref{fig:utilization} we plot the effective utilization of the matchings as the red-dotted curve. As can be seen there, the effective utilization at both $\alpha_1$ and $\alpha_2$ is suboptimal. We propose an algorithm that selects $\alpha$ to maximize effective utilization; a detailed description is deferred to Section~\ref{sec: Algorithm}.

The justification for selecting matchings according to the above is further reinforced by the {\em submodularity} structure of the problem (we discuss submodularity in Section~\ref{sec: Submodularity}).  It turns out that for a certain class of submodular maximization problems with linear packing constraints, greedy algorithms take a form that precisely matches the intuitive thought process above: the proposed intuitively correct algorithm is borne out naturally from submodular combinatorial optimization theory. We briefly recall relevant aspects of submodularity and associated optimization algorithms next. 

\subsection{Submodularity} \label{sec: Submodularity}

A set function $f: 2^{[n]}\rightarrow \mathbb{R}$ is said to be {\em submodular} if it has the following property: for every $A,B\subseteq [n]$ we have $f(A\cup B) + f(A \cap B) \leq f(A) + f(B)$. Alternatively, submodular functions are also defined through the property of {\em decreasing marginal values}: for any $S, T$ such that $ T \subseteq S \subseteq [n]$ and $j \notin S$, we have \begin{align} f(S \cup \{j\}) - f(S) \leq f(T \cup \{j\}) - f(T). \end{align}The difference $f(S \cup \{j\}) - f(S)$ is called the incremental marginal value of element $j$ to set $S$ and is denoted by $f_S(j)$. For our purpose we will only focus on submodular functions that are {\em monotone} and {\em normalized}, i.e., for any $S \subseteq T\subseteq [n]$ we have $f(S) \leq f(T)$ and further $f(\{ \}) = 0$.

Many applications in computer science involve maximizing submodular functions with {\em linear packing constraints}. This refers to problems of the form:
\begin{align}
\max f(S) \quad \text{s.t. } A\mathbf{x}_S \leq b \text{ and } S \subseteq [n]
\end{align}
where $A\in [0,1]^{m\times n}, b\in [1,\infty)^{m}$ and $\mathbf{x}_S$ denotes the characteristic vector of the set $S$. Each of the $A_{ij}$'s is a cost incurred for including element $j$ in the solution. The $b_i$'s represent a total budget constraint. A well-known example of a problem in the above form is the Knapsack problem (the objective function in this case is in fact modular).

With the above background, we  formulate the optimization problem under direct
routing as one of  submodular function maximization. Recall that for any given input traffic matrix $T$, the schedule that is computed is described by a sequence of matchings and corresponding durations. Consider the set $\mathcal{M}$ of all perfect matchings in the complete bipartite graph $K_{n\times n}$ with $n$ nodes in each partite. Then any round in the schedule is simply $(\alpha, P) \in \mathbb{Z}\times \mathcal{M}$. The key observation we make now is to view the schedules as a subset of $\mathbb{Z}\times \mathcal{M}$. Formally,  define a {\em switch schedule} as any subset $\{ (\alpha_1,M_1),\ldots,(\alpha_k,M_k)\}$ of $\mathbb{Z}\times\mathcal{M}$. The objective function in our case is the sum-throughput defined as:
\begin{align}
f(\{ (\alpha_1,M_1),\ldots,(\alpha_k,M_k)\}) = \left\| \min \left\{ \sum_{i=1}^k \alpha_i M_i, T \right\} \right\|_1  \label{eq:func defn}
\end{align}
where the minimum is taken entrywise and $\|\cdot \|_1$ refers to the entrywise $L_1$-norm of the matrix. We observe that the function $f$ is submodular, deferring the proof to the supplementary material. 
\begin{thm} \label{thm: single hop submod}
The function $f:2^{\mathbb{Z}\times \mathcal{M}}\rightarrow \mathbb{R}$ defined by Equation~\eqref{eq:func defn} is a monotone, normalized submodular function.
\end{thm}

We have established that  optical switch scheduling under the sum-throughput metric is a submodular maximization problem. With this, we are ready to  present a greedy algorithm that  achieves a sum-throughput of at least a constant factor of the optimal algorithm for {\em every} instance of the traffic matrix. 

\subsection{Algorithm} \label{sec: Algorithm}

\begin{algorithm}[!tbp]
    \SetKwInOut{Input}{Input}
    \SetKwInOut{Output}{Output}
    \Input{Traffic demand $T$, reconfiguration delay $\delta$ and scheduling window size $W$}
    \Output{Sequence of matchings and weights: $(\alpha_1,P_1),\ldots,(\alpha_k, P_k)$}
$\mathrm{sch} \leftarrow \{ \}$  \tcp*{schedule}
$k \leftarrow 0$ \;
$T_\mathrm{rem} \leftarrow T$ \tcp*{traffic remaining}
\While{$\sum_{i=1}^k (\alpha_i + \delta) \leq W$}
{
$k \leftarrow k + 1$\;
$(\alpha,M) \leftarrow \mathrm{argmax}_{M\in\mathcal{M},\alpha\in\mathbb{R}_+} \frac{ \|\min(\alpha M, T_\mathrm{rem})\|_1}{(\alpha + \delta)}  $ \;
$\mathrm{sch} \leftarrow \mathrm{sch} \cup \{(\alpha, M)\}$ \;
$T_\mathrm{rem} \leftarrow T_\mathrm{rem} - \min(\alpha M, T_\mathrm{rem})$ \;
}
\If{$\sum_{i=1}^k (\alpha_i + \delta) > W$}{$\mathrm{sch} \leftarrow \mathrm{sch} \backslash \{(\alpha, M)\}$\; }
    \caption{Eclipse: greedy direct routing algorithm} \label{alg: sha}
\end{algorithm}

Algorithm~\ref{alg: sha} -- Eclipse -- captures our proposed solution under direct routing. Eclipse takes the traffic matrix $T$, the time window $W$ and reconfiguration delay $\delta$ as inputs, and computes a sequence of matchings and durations as the output. The algorithm proceeds in rounds (the ``while loop"), where in each round a new matching is added to the existing sequence of matchings. The sequence terminates whenever the sum of the matching durations exceeds the allocated time window $W$ or whenever the traffic matrix $T$ is fully covered.

Consider any round $t$ in the algorithm; let $(\alpha_1,M_1),\ldots,$ $(\alpha_{t-1},M_{t-1})$ denote the  schedule computed so far in $t-1$ rounds (stored in variable $\mathtt{sch}$) and let $T_\mathrm{rem}(t)$ denote the amount of traffic yet to be routed. The matching that is selected in the $t$-th round is the one for which {\em utilization} is maximum. Utilization here refers to the percentage of the total matching capacity that is actually used. Mathematically, we choose an $(\alpha,M)$ pair such that $\|\min(\alpha M, T_\mathrm{rem})\|_1$ is maximized. In the supplementary material we have shown that that maxima occurs on the support of $T_\mathrm{rem}$. Hence this can be easily found by looking at the support of the (sparse) matrix $T_\mathrm{rem}$. We also propose a simple binary-search procedure, discussed in Algorithm~\ref{alg: bin search}, that finds a local maxima but performs extremely well in our evaluations (Section~\ref{sec:Evaluations}). 

\begin{algorithm}[!tbp]
    \SetKwInOut{Input}{Input}
    \SetKwInOut{Output}{Output}
    \Input{Traffic demand $T$, reconfiguration delay $\delta$}
    \Output{$(\alpha,M)\in\mathbb{R}_+\times\mathcal{M}$ such that $(\alpha,M) =  \mathrm{argmax}_{M\in\mathcal{M},\alpha\in\mathbb{R}_+} \frac{ \|\min(\alpha M, T)\|_1}{(\alpha + \delta)}$ }
$H \leftarrow$ distinct entries of $T$ sorted in ascending order\;
$i_\mathrm{lb} \leftarrow 1$ and $i_\mathrm{ub}\leftarrow$ length$(H)$\; 	
\While{$i_\mathrm{lb} < i_\mathrm{ub}$}
{
$i \leftarrow (i_\mathrm{lb} + i_\mathrm{ub})/2$ \;
$T_1 \leftarrow \min\{T, H(i) \}$ \tcp*{thresholding $T$ to $H(i)$}
$T_2 \leftarrow \min\{T, H(i+1) \}$ \;
$v_1 \leftarrow (\text{max. weight matching in }T_1)/(H(i)+\delta)$ \;
$v_2 \leftarrow (\text{max. weight matching in }T_2)/(H(i+1)+\delta)$ \;
\uIf{$v_1 < v_2$}{$i_\mathrm{lb}\leftarrow i$\;}
\uElseIf{$v_1 > v_2$}{$i_\mathrm{ub} \leftarrow i$\; }
\Else{return $(H(i),$ max. weight matching in $T_1)$\;}
}
\caption{Finding the greedy maximum} \label{alg: bin search}
\end{algorithm}

As a concluding remark, we note that the constant $\epsilon$ in the approximation factor
comes from the 
requirement that $\alpha + \delta \leq \epsilon W$ hold. We observe that
this mild technical condition, required to show that Eclipse is a constant factor approximation of the optimal algorithm, has an added implication. Informally, it ensures that 
no single matching occupies  the bulk of the scheduling window. This process of selecting a matching is repeated in each round until the sum-duration of the matchings exceed the scheduling window $W$, when the last chosen matching is discarded and the remaining set of matchings are returned. Eclipse is simple
and also fast, a fact the following calculation demonstrates. 

\subsubsection{Complexity}

We begin with the complexity of Algorithm~\ref{alg: bin search}. Since $i_\mathrm{ub}$ is  no more  than the number of distinct entries of $T$, we have $i_\mathrm{ub}\leq n^2$. In each iteration, the algorithm only considers entries of $H$ that have indices between $i_\mathrm{lb}$ and $i_\mathrm{ub}$. However, binary-search halves the effective size of $H$ (i.e., those numbers in $H$ with array indices $i_\mathrm{lb},i_{\mathrm{lb}+1},\ldots,i_\mathrm{ub}$), and the number of iterations of the \texttt{while} loop is bounded by $\log n^2 = 2\log n$. Within the \texttt{while} loop, computing the maximum weight matching can be done in $O(dn^{3/2}\log(W\epsilon))$ time (a basic fact of submodular optimization~\cite{Duan:2012:SAM:2095116.2095227, schrijver-book}) where $dn$ is the number of edges in bipartite graph formed by $T$ (i.e., $d$ is the average sparsity) . Further $(1-\epsilon)$ approximate maximum weight matching can be computed in $O(dn\epsilon^{-1}\log\epsilon^{-1})$~\cite{Duan:2014:LAM:2578041.2529989}, and
efficient implementations in practice have been studied extensively in the literature 
\cite{mekkittikul1998practical,pettie2004simpler,felzenszwalb2011dynamic}. Hence the overall time complexity is $O(dn^{3/2}\log n \log(W\epsilon))$. Now, in Algorithm~\ref{alg: sha} the number of iterations in the \texttt{while} loop is bounded by $W/\delta$. As such the total complexity of the algorithm is $\tilde{O}(dn^{3/2} \frac{W}{\delta} )$. An exact search over the support of $T_\mathrm{rem}$ in the maximization step results in a overall complexity of $\tilde{O}(d^2n^{5/2} \frac{W}{\delta} )$. 

%\begin{algorithm}[htbp]
%    \SetKwInOut{Input}{Input}
%    \SetKwInOut{Output}{Output}
%    \Input{Traffic demand $T$, reconfiguration delay $\delta$ and scheduling window size $W$}
%    \Output{Sequence of matchings and weights: $(\alpha_1,P_1),\ldots,(\alpha_k, P_k)$}
%
%$K \leftarrow$ set of estimates on number of matchings \;
%$v_{\max} \leftarrow 0$\;
%\For{each $k\in K$}{
%$T' \leftarrow$ max. weight $k-$matching of $T$\;
%$(\alpha'_1,P'_1),\ldots,(\alpha'_k,P'_k) \leftarrow$ Single-hop$(T',\delta,W)$\;
%\If{sum-traffic of $(\alpha'_1,P'_1),\ldots,(\alpha'_k,P'_k) > v_{\max}$}{
%$v_{\max} \leftarrow$ sum-traffic of $(\alpha'_1,P'_1),\ldots,(\alpha'_k,P'_k)$\;
%$(\alpha_1,P_1),\ldots,(\alpha_k,P_k) \leftarrow (\alpha'_1,P'_1),\ldots,(\alpha'_k,P'_k)$
%}
%}
%    \caption{Single-hop - Sparsifier}
%\end{algorithm}

%\subsubsection{Example}

%For clarity, we illustrate the proposed algorithm with an example. Consider traffic matrix T %given

\subsubsection{Approximation Guarantee}

Since the  proposed direct routing algorithm is connected to 
 submodular maximization with   linear constraints, we can adapt standard combinatorial optimization techniques to show an approximation factor of $1-1/e$. Let \texttt{OPT} denote the sum-throughput of the optimal algorithm for given inputs $T,\delta$ and $W$. Let \texttt{ALG2} denote the sum-throughput achieved by Eclipse. We then have the following.
\begin{thm} \label{thm: submod apx guarantee}
If the entries of $T$ are bounded by $\epsilon W +\delta$ then Eclipse approximates the optimal algorithm to within a factor of $1-1/e^{(1-\epsilon)}$, i.e.,
\begin{align}
\mathtt{ALG2} \geq (1-1/e^{(1-\epsilon)})\mathtt{OPT}.
\end{align}
\end{thm}
The proof of the above Theorem is deferred to the Appendix.

%
%\begin{figure}[!htbp]
%  \centerline{\includegraphics[height=60mm]{figures/plot1.pdf}}
%    \caption{Reachability of nodes under multi-hop routing.}
%    \label{fig:plot1}
%\end{figure}
%
%\begin{figure}[!htbp]
%  \centerline{\includegraphics[height=60mm]{figures/plot2.pdf}}
%    \caption{Reachability of nodes under multi-hop routing.}
%    \label{fig:plot2}
%\end{figure}
%
%\begin{figure}[!htbp]
%  \centerline{\includegraphics[height=60mm]{figures/plot3.pdf}}
%    \caption{Reachability of nodes under multi-hop routing.}
%    \label{fig:plot3}
%\end{figure}
%
%\begin{figure}[!htbp]
%  \centerline{\includegraphics[height=60mm]{figures/plot4.pdf}}
%    \caption{Reachability of nodes under multi-hop routing.}
%    \label{fig:plot4}
%\end{figure}

\section{Indirect Routing} \label{sec: Indirect Routing}

In the previous section, we focused on direct routing where packets are forwarded to their destination ports only if a link {\em directly} connecting the source port to the destination port appeared in the schedule -- this is essentially a ``single-hop'' protocol. In this section, we explore  allowing packets to be forwarded to (potentially) multiple intermediate ports before arriving at its final destination.  In terms of implementing this more-involved protocol, we note that there is no extra overhead needed: the destination of any received packet is read first upon reception and since the queues are maintained on a  per-destination basis at each ToR port, any received packet can be diverted to the appropriate queue. The key point of allowing indirect routing is the vastly increased range of ports that can be reached from a small number of matchings. 

Consider Fig.~\ref{fig:multiflow} which illustrates a $6$-port network and a sequence of $3$ consecutive matchings in the schedule. In the first matching port 3 is connected to port 2; in the second matching it is connected to port 5 and so on. With direct routing, port 3 can only forward packets to port 2 in round 1, port 5 in round 2 and port 1 in round 3. In other words the set of egress ports {\em reachable} by port 3 is $\{2,3,5\}$. In the indirect routing framework of this section, port 3 can also forward packets to port 1. This can be achieved by first forwarding the packets to port 2 in the first round where the packets are queued. Then in the second round we let port 2 forward those packets to the destination port 3. Thus the {\em reachability} of the nodes is enhanced by allowing for indirect routing. Indirect routing can also be viewed as ``multi-hop'' routing.    

Traditionally multi-hop routing has been used as a means of {\em load balancing}. This is known to be true in the context of networks  such as the Internet where the benefits of ``Valiant load-balancing" are legion \cite{rabin1989efficient,valiant1990bridging,greenberg2008towards}. The benefits of load balancing are also well known in the switching context -- a classic example is the two-stage load-balancing algorithm in crossbar switches without reconfiguration delay~\cite{Srikant:2014:CNO:2636796}. The benefit of multi-hop routing in our context is {\em markedly different}: the reachability benefits of indirect routing are especially well suited to the setting where input ports are directly connected to only a few output ports due to the {\em small number} of matchings in the scheduling window. In fact, an elementary calculations shows that over a period of $k$ matchings in the schedule, indirect routing can allow a node to forward packets to $O(2^k)$ other nodes, compared to only $O(k)$ nodes possible with direct routing. This is because of the recursion $f(k) = 2f(k-1) + 1$ where $f(k)$ denotes the number of nodes reachable by any node in $k$ rounds. If a node (say, node 1) can reach $f(k-1)$ nodes in $k-1$ rounds, then in the $k$th round (i) there is a new node directly connected to node 1 and (ii) each of the $f(k-1)$ nodes can be connected to a new node. Thus the number of nodes connected to node 1 in the $k$ -th round becomes $f(k-1) + (f(k-1) + 1)$. Fig.~\ref{fig:multiflow} illustrates this phenomenon where reachability from node 6 is shown. 
\begin{figure}[!tbp]
  \centerline{\includegraphics[height=50mm]{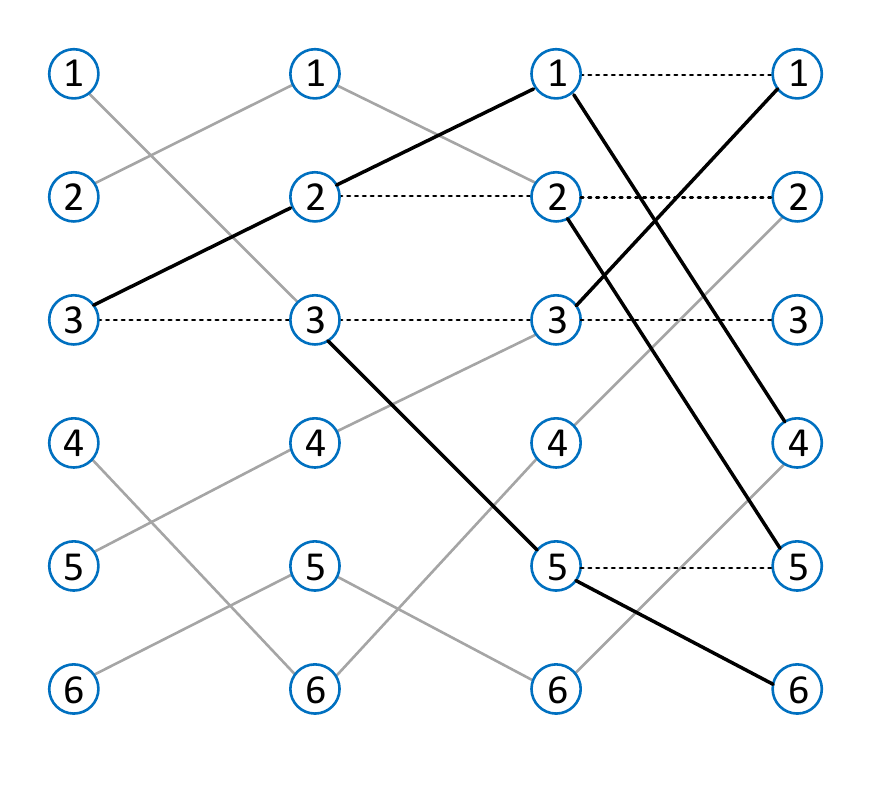}
  \label{fig:multiflow}}
    \caption{Reachability of nodes under multi-hop routing.}
    \label{fig:multiflow}
\end{figure}
As a corollary we observe that $O(\log_2 n)$ rounds of matchings are sufficient to reach {\em all} other nodes in a $n$-port network. 

As in the direct routing case, computing the optimal schedule remains a challenging problem. While it is clear that we can achieve a performance at least as good as with direct routing, the gain is different for each instance of the traffic matrix -- precisely quantifying the gain in an instance-specific way appears to be challenging.  Our main result is that the submodularity property of the objective function continues to hold, provided the variables are considered in the appropriate format. Further, if we restrict some of the variables (the matchings and their durations), then there is also a natural, simple and fast greedy algorithm to compute the switching schedules and routing policies that is  approximately optimal for {\em each} instance of the traffic matrix. This algorithm serves as a heuristic solution to the more general problem of jointly finding number of switchings, their durations, the switching schedules and routing policies. We present these results, following the same format as in the direct routing section, leaving numerical evaluations to a later section. We follow the model as discussed in Section~\ref{sec: System Model}. 

%\subsection{Model}
%Following the same model as earlier (section~\ref{sec: System Model}), we make the extra supposition  that packets received at the egress port get buffered in VOQs in each round for (re)transmission in subsequent rounds. 
%
%\subsubsection*{Objective} \label{sec: objective multi hop}
%For a given traffic demand matrix $T$, we need to  compute a sequence of matchings, durations $(\alpha_1,M_1),\ldots,(\alpha_k,M_k)$ such that $\sum_{i=1}^k\alpha_i+k\delta \leq W$, $M_i\in\mathcal{M}$ for all $i\in[k]$ and the total amount of traffic sent is maximized. With indirect routing, in addition to specifying the sequence of matchings and their durations, we also need to specify how the (multi-hop) traffic is routed across the rounds.  

\subsection{Submodularity of objective function}
We first adopt an alternative way of describing the switch schedule. Instead of specifying switch configurations round by round, and then specifying a multi-hop routing strategy in each round, we directly specify the {\em path} taken by each packet. The path-based formulation of multi-hop routing is well known in the problems of computing maximal flows in capacitated graphs and is crucial to understanding flow-cut gaps \cite{williamson2011design}; such a formulation serves us well in the {\em causal} structure of the routed traffic patterns that naturally occur here. 

For simplicity we fix the number of rounds $k$ in the schedule. Consider a fully connected $k$-round time layered graph $G$: this graph consists of $k+1$ partites, $V_0,V_1,\ldots,V_k$, of $n$ nodes each. Nodes in each partite $i$ have directed edges to partite $i+1$ (such that the two partites form a complete bipartite graph).  Let $\mathcal{P}$ denote the set of all paths in $G$ that begin at a node in $V_0$ and end at a node in $V_k$. Now, if we want to describe a multi-hop route for a packet in the system we can do it by choosing a $p\in\mathcal{P}$. If we are able to choose a path for every packet in the traffic matrix $T$, such that the union of the paths obey capacity constraints, then we would have succesfully specified a sequence of switch configurations and a routing policy for the schedule. Now, consider a function $f:2^{\mathbb{Z}\times\mathcal{P}}\rightarrow \mathbb{Z}$ defined by $f(\{(\beta_1,p_1),\ldots,(\beta_m,p_m)\})  \triangleq$
\begin{align}
\sum_{i,j\in[n]}\min\left(\sum_{l=1}^m \beta_l\mathbf{1}_{\substack{p_l(0)=i, \\ p_l(k+1)=j}} , T_{ij}\right) \label{eq: multhop submod}
\end{align} 
where $p(0)$ and $p(k+1)$ denote the starting and ending nodes of path $p$, and $\mathbf{1}_{\{\cdot\}}$ is the indicator function. Then the key observation is that $f$ is submodular.  
\begin{thm}
The function $f:2^{\mathbb{Z}\times\mathcal{P}}\rightarrow \mathbb{Z}$ defined by Equation~\eqref{eq: multhop submod} is submodular. 
\end{thm}
The proof is similar to Theorem~\ref{thm: single hop submod} and we omit it in the interests of space. So far we have not imposed any restrictions on the set of paths that we choose for the schedule. This can be incorporated in the form of constraints to the problem. Thus we are able to rephrase the objective as a constrained submodular maximization problem. 

{\bf Constraints}: 
Since we are only choosing weighted paths, we need to ensure that (i) the set of paths form a matching in each round and (ii) the total durations of the matching is at most $W - k\delta$. As such, we can write the following constraints for any subset $\{(\beta_1,p_1),\ldots,(\beta_m,p_m)\}\in 2^{\mathbb{Z}\times \mathcal{P}}$ :   
\begin{align}
\sum_{\substack{e: v\in e, \\ e\in E_j}}\mathbf{1}\left\{ \sum_{i=1}^m\mathbf{1}_{\{e\in p_i\}}\beta_i > 0\right\} \leq 1 ~\forall v\in V_{j-1}, j\in[k] \label{eq: submod const 1}\\
\sum_{\substack{e: v\in e, \\ e\in E_j}}\mathbf{1}\left\{ \sum_{i=1}^m\mathbf{1}_{\{e\in p_i\}}\beta_i > 0\right\} \leq 1 ~\forall v\in V_{j}, j\in[k] \label{eq: submod const 2} \\
\sum_{j=1}^k \left(\left( \max_{e\in E_j} \sum_{i=1}^m \mathbf{1}_{\{e\in p_i\}} \beta_i  \right)+ \delta\right) \leq W  \label{eq: submod const 3}
\end{align}
where $E_j$ stands for the edges between $V_{j-1}$ and $V_j$ in $G$. Hence we can express the problem as maximization of objective~\eqref{eq: multhop submod} subject to the constraints~\eqref{eq: submod const 1}--~\eqref{eq: submod const 3}.

The key challenge here is that  the constraints are {\em nonlinear} --  it is not clear whether an efficient (approximation) algorithm exists. The nonlinearities appear only in the sense of membership tests and a corresponding thresholding function -- so it is possible that an efficient nearly-optimal greedy algorithm exists; such a study is outside the scope of this manuscript.  We do note, however, that for the special case in which the configurations are fixed and we only have to decide on the indirect routing policies, the constraints take on a linear form -- in this setting, we are able to construct fast and efficient greedy algorithms. This case represents a composition of  direct routing (where switch schedules are computed) and indirect routing (where the multi-hop routing policies are described), and is discussed next. 

{\bf Multi-Hop Routing Policies}: 
Consider a fixed sequence $(\alpha_1,M_1),\ldots,(\alpha_k,M_k)$ of switch configurations and an input traffic demand matrix $T$. Let $G$ denote the time-layered graph obtained from the sequence of matchings, i.e., $G$ consists of $k+1$ partites $V_0,\ldots,V_k$ with $n$ nodes each, and $M_i$ is the matching between partites $V_{i-1}$ and $V_i$ with edge weight $\alpha_i$ on the matching edges. In addition to the matching edges, there are also edges, with infinite edge weights, connecting the $j$-th nodes of $V_{i-1}$ and $V_i$ for all $j\in[n],i\in[k]$. In other words, $G$ is a time-layered graph whose edges are weighed according to the capacity available on the edges. Let $R(e)$ denote the capacity (= edge-weight) of edge $e\in G$. In this setting,  the capacity constraints on the end-to-end paths are the sole constraints left in the optimization problem: in other words, the constraints in Equations~\ref{eq: submod const 1} and~\ref{eq: submod const 2} simplify to the following -- we  consider subsets $\{(\beta_1,p_1),\ldots,(\beta_m,p_m)\}$ that obey: 
\begin{align}
\sum_{i=1}^m \beta_i\mathbf{1}_{\{e\in p_i\}} \leq R_j(e) ~~\forall e\in R_j, \forall j\in[k] \label{eq: submod lin constraints}
\end{align}  
Notice that the constraints above have a linear form, and there are a total of $kn$ such constraints (one for each edge). Such a setting allows for a natural, simple, fast and nearly-optimal algorithm which we discuss below. Prior to that discussion, we remark that a naive approach to resolve the setting here is to formulate a {\em linear program}  that maximizes the required objective. Indeed linear programming based approaches was the predominant technique used to solve this classical {\em multicommodity flow} problem~\cite{ahuja1993network, grötschel1993geometric, barnhart1993network}. However, despite many years of research in this direction the proposed algorithms were often too slow even for moderate sized instances~\cite{leighton1995fast}. Since then there has been a renewed effort in providing efficient {\em approximate} solutions to the multicommodity flow problem~\cite{garg2007faster, arora2012multiplicative}. The algorithm we propose is also a step in this direction, favoring efficiency over exactness of the solution. Note that in the path-formulation of the linear program there can be an exponential (in $k$) number of variables (for example, a schedule where node 1 is connected to node 2 and vice-versa in all the $k$ matchings has an exponential number of indirect paths from node 1 to 2).  Thus there is no obvious efficient (exact) solution to this LP. On the other hand, the formulation we work with  is able to handle this issue by appropriately {\em weighing the edges} and {\em picking the path} with the smallest weight; this latter step  can be done efficiently -- for example, using Dijkstra's algorithm. 

 % The trouble, however, that this linear program can have an {\em exponential} (in $k$) number of variables (for example, a schedule where node 1 is connected to node 2 and vice-versa in all the $k$ matchings has an exponential number of indirect paths from node 1 to 2).  Thus there is no obvious efficient (exact) solution to this LP. On the other hand, the formulation we work with  is able to handle this issue by appropriately {\em weighing the edges} and {\em picking the path} with the smallest weight; this latter step  can be done efficiently -- for example, using Dijkstra's algorithm.

\subsection{Algorithm}

We propose an algorithm Eclipse++ directly motivated from~\cite{azar2012efficient}, which presents a fast and efficient multiplicative weights algorithm for submodular maximization under linear constraints. The pseudocode has been given in Algorithm~\ref{alg: mha}. The structure of Algorithm~\ref{alg: mha} is  similar in spirit to direct-routing Algorithm~\ref{alg: sha} in the sense that (a) the algorithm proceeds in rounds, where one new path is added to the schedule in each round and (b) we select a path that offers the greatest utility per unit of cost incurred. However, unlike Algorithm~\ref{alg: sha} where there was only one linear constraint, we have multiple linear constraints now. This is addressed by assigning weights to the constraints and considering a {\em linear combination} of the costs as the true cost in each round. We describe the salient features of Algorithm~\ref{alg: mha} below. 
\begin{algorithm}[!tbp]
    \SetKwInOut{Input}{Input}
    \SetKwInOut{Output}{Output}
    \Input{Traffic demand $T$, switch configurations with residue capacities $R_1,\ldots,R_k$, update factor $\lambda$}
    \Output{Sequence of paths and weights: $(\beta_1,p_1),\ldots,(\beta_m, p_m)$}
$\mathrm{sch} \leftarrow \{ \}$  \tcp*{schedule} 

$T_\mathrm{rem} \leftarrow T$ \tcp*{traffic remaining} 
$w_e \leftarrow 1/R(e)$ for all $e\in E$\;
$m\leftarrow 1$\;
\While{$\sum_{e\in E} R(e) w_e\leq \lambda$ and $\|T_\mathrm{rem}\|_1 > 0$}
{
$(\beta_m,p_m) \leftarrow \mathrm{argmax}_{p\in\mathcal{P},\beta\in\mathbb{Z}} \frac{ \min(\beta, T_\mathrm{rem}(p(0,p(k+1))}{\sum_{e\in E}\beta \mathbf{1}_{\{e\in p\}}w_e}  $ \; 
$\mathrm{sch} \leftarrow \mathrm{sch} \cup \{(\beta_m, p_m)\}$ \;
$T_\mathrm{rem}(p_m(0),p_m(k+1)) \leftarrow T_\mathrm{rem}(p_m(0),p_m(k+1)) -\beta$ \;
$w_e \leftarrow  w_e \lambda^{\beta_m \mathbf{1}_{\{e\in p\}}/R(e)}~~\forall e\in G$ \;
$m\leftarrow m+1$\;
}
\uIf{$\sum_{i=1}^{m-1} \beta_i\mathbf{1}_{\{e\in p_i\}} \leq R(e) ~~\forall e\in E$}{return sch}
\Else{return sch$\backslash (\beta_{m-1},p_{m-1})$}
    \caption{Eclipse++ : greedy indirect routing algorithm} \label{alg: mha}
\end{algorithm}

We begin by recalling that we have one constraint for each edge in the matchings for a total of $kn$ constraints. Let $w_e$ denote the weight assigned to the constraint involving edge $e$ and let $R(e)$ denote the capacity available with edge $e$. We set $w_e = 1/R(e)$ for all $e$ initially, i.e., edges with a large capacity are assigned a small weight and vice-versa. Thus in addition to the time-layered capacity graph $G$, we can now have another graph $G_w$ (with same topology as $G$) whose edges are weighted by $w_e$. Now, for any path $p$ the ``effective cost'' of the path per unit of flow is simply the total cost of $p$ in $G_w$. Thus for the path $(\beta,p)$ carrying $\beta$ units of flow, the effective cost is given by $\sum_{e\in E}\beta w_e \mathbf{1}_{e\in p}$. On the other hand, the benefit we get due to adding path $(\beta, p)$ is given by $\min (\beta, T(p(0),p(k+1)))$ where $p(0)$ and $p(k+1)$ stand for the starting and terminating nodes along path $p$. Thus, the ratio $\frac{\min (\beta, T(p(0),p(k+1)))}{\sum_{e\in E}\beta w_e \mathbf{1}_{e\in p}}$ denotes the benefit of path $p$ per unit cost incurred. In Algorithm~\ref{alg: mha} we select $p$ such that the utility per unit cost is maximized.   

Now, once we have selected a path $(\beta_1,p_1)$ in the first round, we update the weights $w_e$ on the edges. To do this we let a parameter $\lambda$ be input to the algorithm. Then, for each edge $e$ along the path $p$ the weights are updated as $w_e \leftarrow w_e \lambda^{\beta_1/R(e)}$; for the remaining edges the weights remain unchanged. Thus repeating the above iteratively until the \texttt{while} loop condition $\sum_{e\in E} R(e) w_e \leq \lambda$ becomes invalid, we get a schedule that is the output of the algorithm. It can also be shown that if the schedule returned \texttt{sch} violates any of the constraints (Equation~\eqref{eq: submod lin constraints}) then it must have happened at the very last iteration and hence we return a schedule with the last added path removed from it.  A detailed analysis and correctness of the algorithm proposed is deferred to a full version of the paper and is omitted in this conference submission. %We do not elaborate further, and refer to~\cite{azar2012efficient} for a more detailed analysis and correctness of the general maximization algorithm. 

It only remains to show how the maximizer of 
\begin{align}
\frac{ \min(\beta, T_\mathrm{rem}(p(0,p(k+1))}{\sum_{e\in E}\beta \mathbf{1}_{\{e\in p\}}w_e} \label{eq: submod maximizer}
\end{align} is computed efficiently in each round (first line inside the \texttt{while} loop). Recall that $G_w$ denotes the time-layered graph $G$ with edges weighted by $w_e$. Consider the set of shortest paths in $G_w$ (= smallest $w_e$-weighted path) from  vertices in $V_0$ to vertices in $V_k$. Let $p^*$ denote the shortest among them. Then by setting $\beta^* \leftarrow T_\mathrm{rem}(p^*(0),p^*(k+1))$ we claim that Equation~\eqref{eq: submod maximizer} is maximized. This is because, 
\begin{align}
\frac{ \min(\beta, T_\mathrm{rem}(p(0,p(k+1))}{\sum_{e\in E}\beta \mathbf{1}_{\{e\in p\}}w_e} \leq \frac{ \beta }{\sum_{e\in E}\beta \mathbf{1}_{\{e\in p\}}w_e} \\
\leq \frac{1}{\min \sum_{e\in E} \mathbf{1}_{\{e\in p\}}w_e}. 
\end{align}
If $T_\mathrm{rem}(p^*(0),p^*(k+1)) = 0$ we proceed to the second smallest shortest path and so on. This allows a very efficient implementation of the internal maximization step. 

\subsubsection{Approximation Guarantee}

We show, as in the direct-routing scenario, that Eclipse++ has a constant factor approximation guarantee. Specifically, for a fixed instance of the traffic matrix,  let \texttt{OPT} and \texttt{ALG4} denote the value of the objectives achieved by the optimal algorithm and Eclipse++, respectively. Let $\eta := \max_{i,j\in[n],e\in E} T(i,j) /R(e)$. Then one can show that $\mathtt{ALG4} = \Omega(1/(nk)^\eta) \mathtt{OPT}$ for $\lambda=e^{1/\eta}nk$; the proof is analogous to the direct-routing case, it follows \cite[Theorem 1.1]{azar2012efficient},  and is deferred to a  full version of this manuscript. Further, if $\eta = O(\epsilon^2/\log (nk))$ for some fixed $\epsilon > 0$ then we get a approximation ratio of $(1-\epsilon)(1-1/e)$ by letting $\lambda = e^{\epsilon/(4\eta)}$ (this observation is inspired by \cite[Theorem 1.2]{azar2012efficient}). An interesting regime where this setting occurs is when the traffic matrices are dense with small skew. For example, we get a constant factor approximation if the sparsity of the traffic matrix grows logarithmically fast (or faster). This is in stark contrast to direct routing, where sparse matrices generally perform better.  

\subsubsection{Complexity}
The proposed algorithm is simple and fast. In this subsection, we explicitly enumerate the time complexity of the full algorithm and show that the complexity is at most cubic in $n$ and nearly linear in $k$.  Let $W$ denote a bound on the total incoming or outgoing traffic for a node. In each iteration of the \texttt{while} loop at least one unit of traffic is sent. Therefore there are at most $W$ iterations of the \texttt{while} loop. Now, in each iteration finding the shortest paths between nodes in $V_0$ to nodes in $V_k$ takes $kn^2(\log k + \log n)$ operations using Dijkstra's algorithm~\cite{fredman1987fibonacci}. Sorting the computed distances takes $kn^2(\log k + \log n)^2$ time and at most $n^2$ more operations to find a pair $i,j$ such that $T_\mathrm{rem}(i,j)>0$. Finally the weights update step takes $kn$ time. Therefore overall it takes $O(kn^2(\log k + \log n)^2)$ time per iteration. Hence the time complexity of the complete algorithm is $O(Wkn^3(\log k + \log n)^2)$.

\section{Evaluation} \label{sec:Evaluations}

\begin{figure*}[t]
  \centerline{\subfigure[]{\includegraphics[height=50mm]{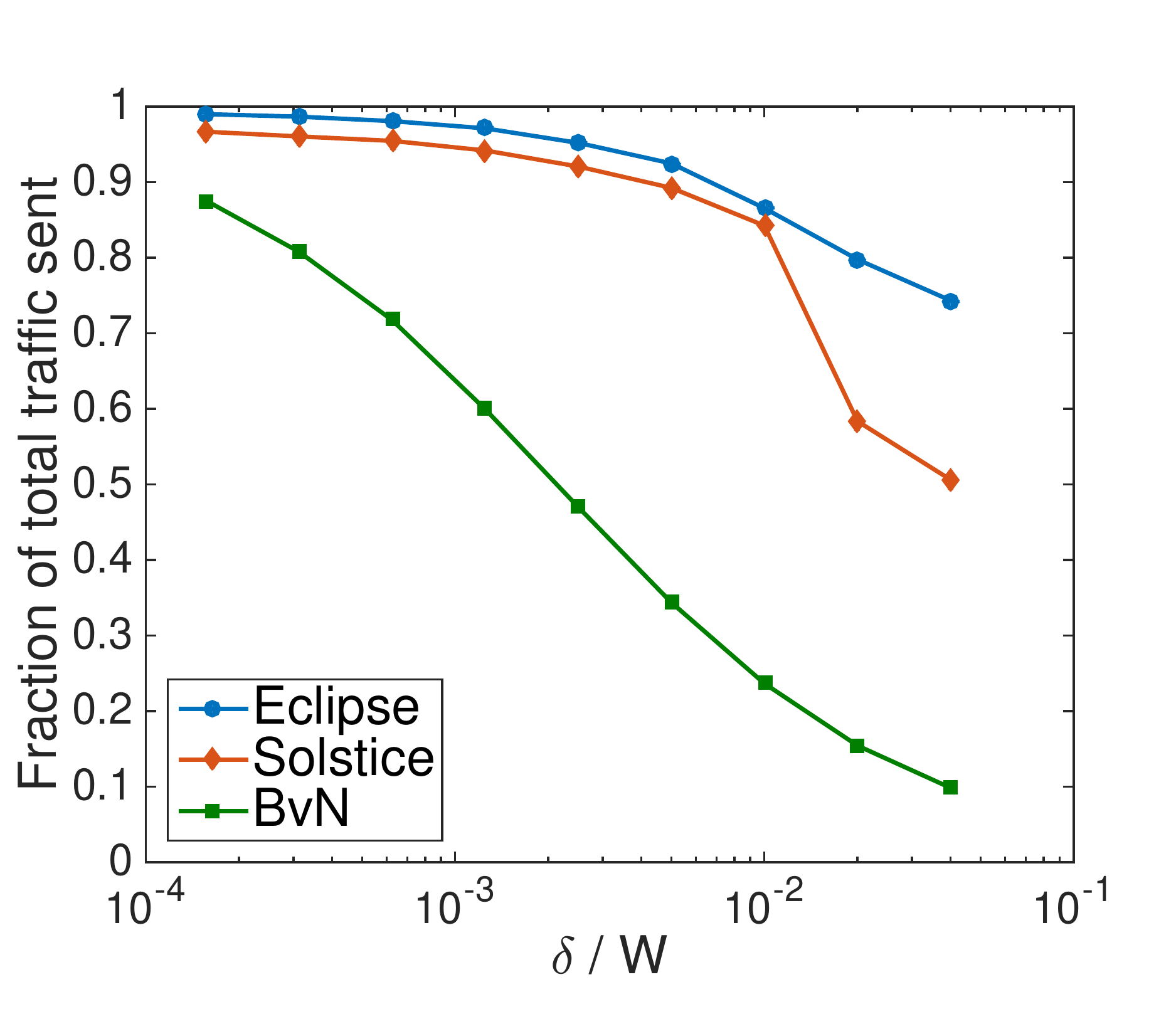}
     \label{fig:plot1}}
\hfil
     \subfigure[]{\includegraphics[height=50mm]{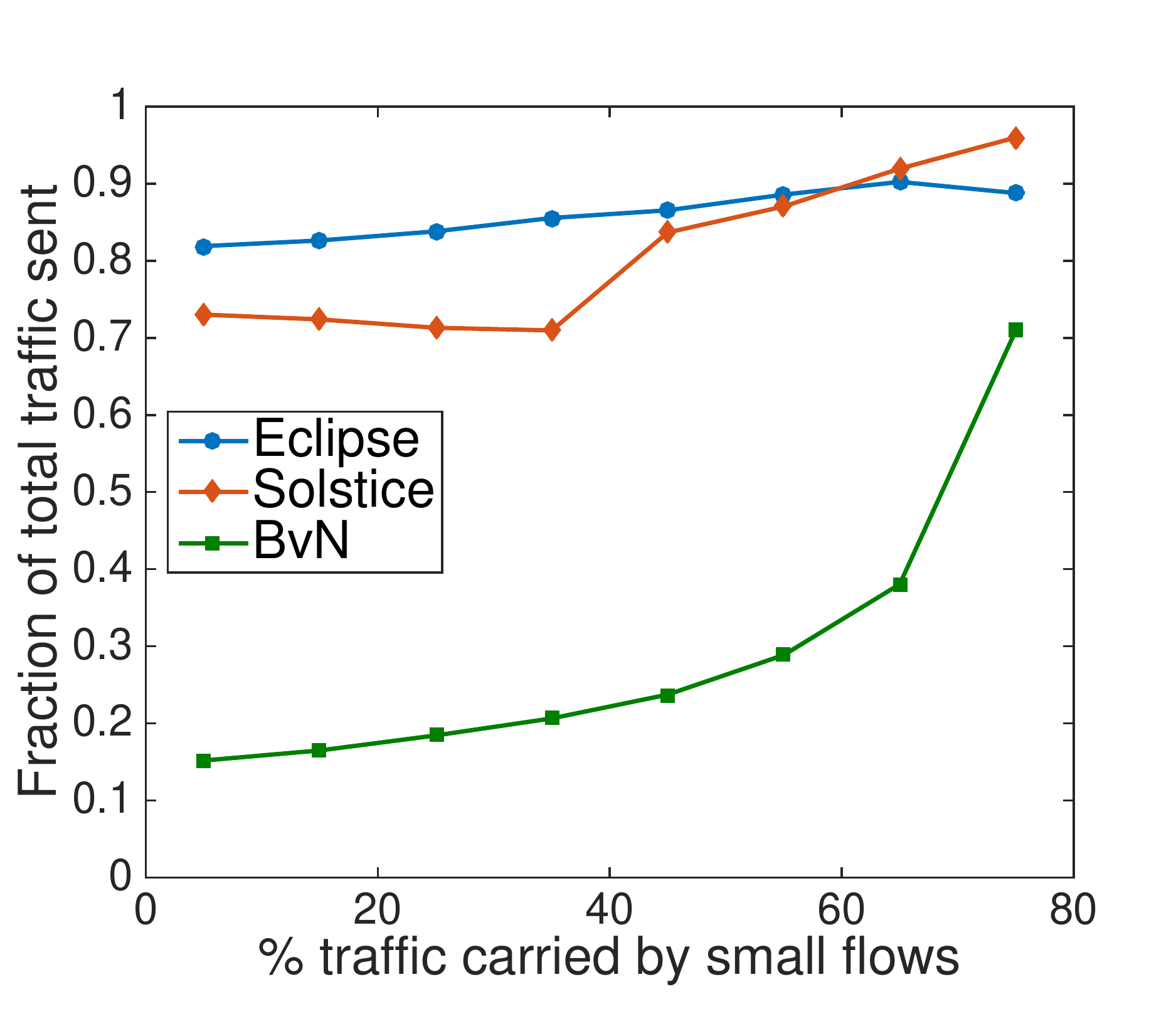}
     \label{fig:plot2}}
\hfil
     \subfigure[]{\includegraphics[height=50mm]{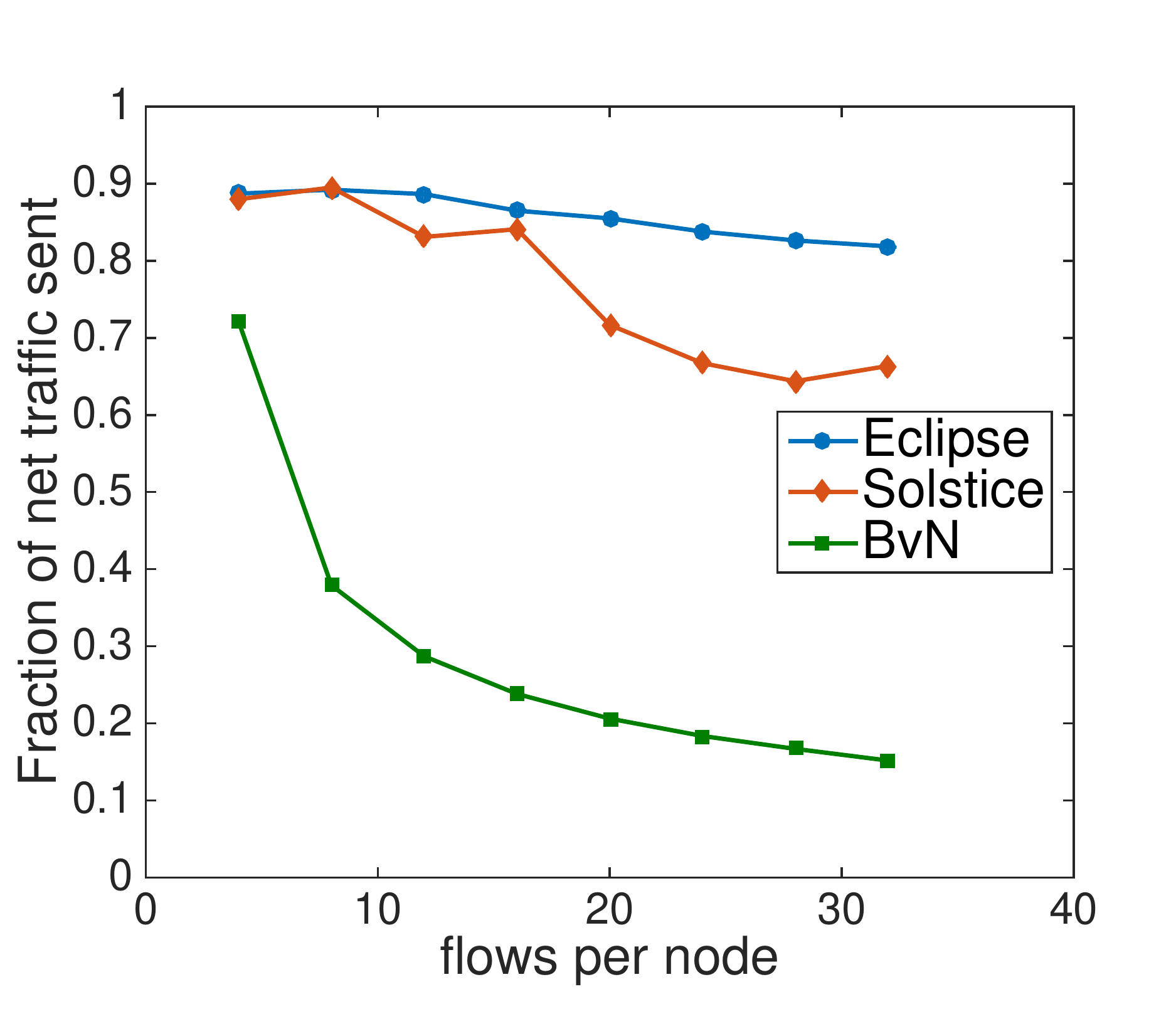}
     \label{fig:plot3}}}
    \caption{Performance comparison of Eclipse under single-block inputs.}
    \label{fig:plots1-3}
\end{figure*}

In this section, we compliment our analytical results with numerical
simulations to explore the effectiveness of our algorithms and compare
them to state-of-the-art techniques in the literature. We empirically
evaluate both the direct routing algorithm (Eclipse;
Algorithm~\ref{alg: sha}) and the indirect routing algorithm
(Eclipse++; Algorithm~\ref{alg: mha}).

\smallskip
\noindent{\bf Metric:}
We consider the total fraction of traffic delivered via the circuit
switch (sum-throughput) as the performance metric throughout this
section.

\smallskip
\noindent{\bf Schemes compared.}
Our experiments compare Eclipse against two existing algorithms for
direct routing:

\noindent{\em Solstice~\cite{Solstice}:} This is the state-of-the-art
hybrid circuit/packet scheduling algorithm for data centers. The key
idea in Solstice is to choose matchings with 100\% utilization. This
is achieved by thresholding the demand matrix and selecting a perfect
matching in each round. The algorithm presented in~\cite{Solstice}
tries to minimize the total duration of the window such that the
entire traffic demand is covered. In this paper, we have considered a
more general setting where the scheduling window $W$ is constrained.
To compare against Solstice in this setting, we truncate its output
once the total schedule duration exceeds $W$.

\noindent{\em Truncated Birkhoff-von-Neumann (BvN)
  decomposition~\cite{chang2000birkhoff}:} The second algorithm we
compare against is the truncated BvN decomposition algorithm. BvN
decomposition refers to expressing a doubly stochastic matrix as a
convex combination of permutation matrices and this decomposition
procedure has been extensively used in the context of packet switch
scheduling~\cite{Solstice,keslassy2003guaranteed,chang2000birkhoff}.
However BvN decomposition is oblivious to reconfiguration delay and
can produce a potentially large ($O(n^2)$) number of matchings.
Indeed, in our simulations BvN performs poorly.

Indirect routing is relatively new and to the best of our knowledge
our work is the first to consider use of indirect routing for
centralized scheduling.\footnote{Indirect routing in a distributed
  setting but without consideration of switch reconfiguration delay
  was studied in a recent work \cite{cao2014joint}.} In our second set
of simulations, we show that the benefits of indirect routing are in
addition to the ones obtained from switch configurations scheduling.
To this end, we compare Eclipse with Eclipse++ to quantify the
additional throughput obtained by performing indirect routing
(Algorithm~\ref{alg: mha}) on a schedule that has been (pre)computed
using Eclipse.

\smallskip
\noindent{\bf Traffic demands:}
We consider two classes of inputs: (a) single-block inputs and (b)
multi-block inputs (explained in Section~\ref{sec: DR
  setup}). Intuitively, single-block inputs are matrices which consist
of one $n\times n$ `block' that is sparse and skewed, and are similar
to the traffic demands evaluated in the Solstice
paper~\cite{Solstice}. Multi-block inputs, on the other hand, denote
traffic matrices that are composed of many sub-matrices each with
disparate properties such as sparsity and skew.

\smallskip
\noindent{\bf Network size:}
The number of ports is fixed in the range of 50--200. We find that the
relative performances stayed numerically stable over this range as
well as for increased number of ports.

%Our observation is Single-block inputs are a class of inputs considered in Solstice.
% Our main numerical observation is that Solstice has a performance
% inferior, but still somewhat comparable, to Eclipse under single-block
% inputs (Fig.~\ref{fig:plots1-3}). However, with multi-block inputs we
% see that Eclipse has a significantly better performance than Solstice
% (sometimes more than double better); Fig.~\ref{fig:plots4-6}).

\subsection{Direct Routing}
While maintaining the sum-throughput as the performance metric, we
vary the various parameters of the system model to gauge the
performance in different situations.

% In particular, the parameters we
% consider are: (a) reconfiguration delay; (b) sparsity; (c) skewness;
% (d) multi-block inputs on Eclipse.

\subsubsection*{Single-block Inputs} \label{sec: DR setup}

For a single-block input, our simulation setup consists of a network with $100$ ports. The link rate of the circuit switch is normalized to $1$, and the scheduling window length is also 1 $(W=1)$. We consider traffic inputs where the maximum traffic to or from any port is bounded by $W$. Further, we let the reconfiguration delay $\delta = W / 100$. The traffic matrix is generated similar to~\cite{Solstice} as follows. We assume $4$ large flows and $12$ small flows to each input or output port. The large flows are assumed to carry $70\%$ of the link bandwidth, while the small flows deliver the remaining $30\%$ of the traffic. To do this, we let each flow be represented by a random weighted permutation matrix, i.e., we have
\begin{align}
T = \sum_{i=1}^{n_L} \frac{c_L}{n_L} P_i +  \sum_{i'=1}^{n_S} \frac{c_S}{n_S} P_{i'} + N \label{eq: simulation model}
\end{align}
where $n_L (n_S)$ denotes the number of large (small) flows and
$c_L(c_S)$ denotes the total percentage of traffic carried by the
large (small) flows. In this case, we have $n_L=4,n_S=12$ and
$c_L=0.7, c_S=0.3$. Further, we have added a small amount of noise $N$
--- additive Gaussian noise with standard deviation equal to $0.3\%$
of the link capacity --- to the non-zero entries to introduce some
perturbation. Each experiment below has been repeated 25 times.

\noindent
\textbf{Reconfiguration delay:} In Fig.~\ref{fig:plot1} we plot
sum-throughput while varying the reconfiguration delay from $W/3200$
to $4W/100$. Ecliplse achieves a throughput of at least 90\% for
$\delta \leq W/100$. We observe Eclipse to be consistently better than
Solstice although the difference is not pronounced until
$\delta > W/100$. The BvN decomposition algorithm has a large
throughput when the reconfiguration delay is small. As $\delta$
increases, it's performance gradually worsens.

\noindent
\textbf{Skew:} We control the skew by varying the ratio of the amount
of traffic carried by small and large flows in the input traffic
demand matrix ($c_L/c_S$ in Equation~\eqref{eq: simulation model}).
Fig.~\ref{fig:plot2} captures the scenario where the percentage
traffic carried by the small flows is varied from 5 to 75. We observe
that Eclipse is very robust to skew variations and is able to
consistently maintain a throughput of about 85\%. Solstice has a
slightly better performance at low skew (when small-flows carry
$\sim 75\%$ of traffic); but overall, is dominated by Eclipse.

\noindent
\textbf{Sparsity:} Finally, we tested the algorithms' dependence on
sparsity and plotted the results in Fig.~\ref{fig:plot3}. The total
number of flows is varied from 4 to 32, while fixing the ratio of the
number of large to small flows at 1:3. As the input matrix becomes
less sparse, the performance of algorithms degrades as expected.
However, for Eclipse, the reduction in the throughput is never more
than 10\% over the range of sparsity parameters considered. Solstice,
on the other hand, is affected much more severely by decreased
sparsity.

\begin{figure*}[t]
  \centerline{\subfigure[]{\includegraphics[height=50mm]{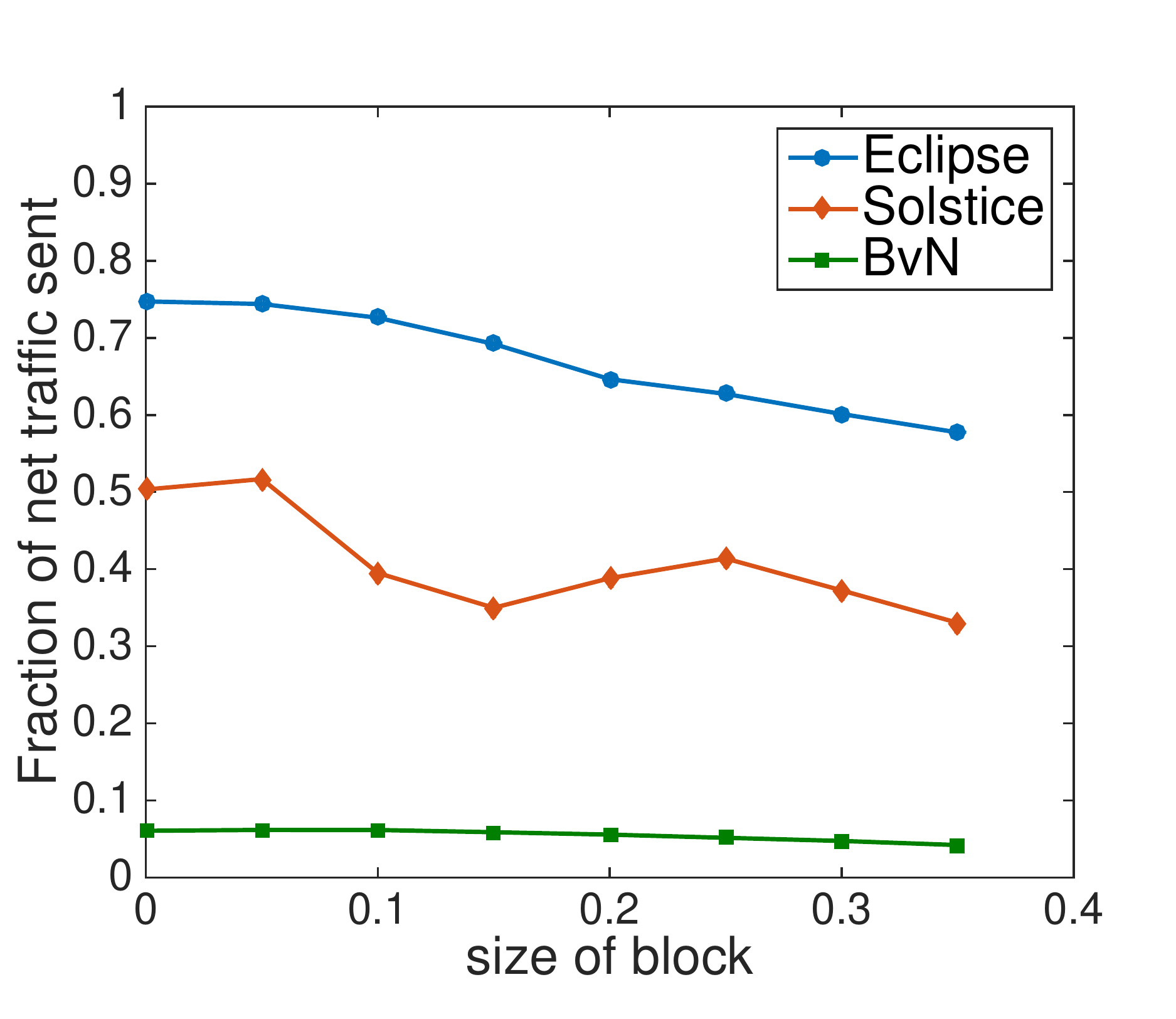}
     \label{fig:plot4}}
\hfil
     \subfigure[]{\includegraphics[height=50mm]{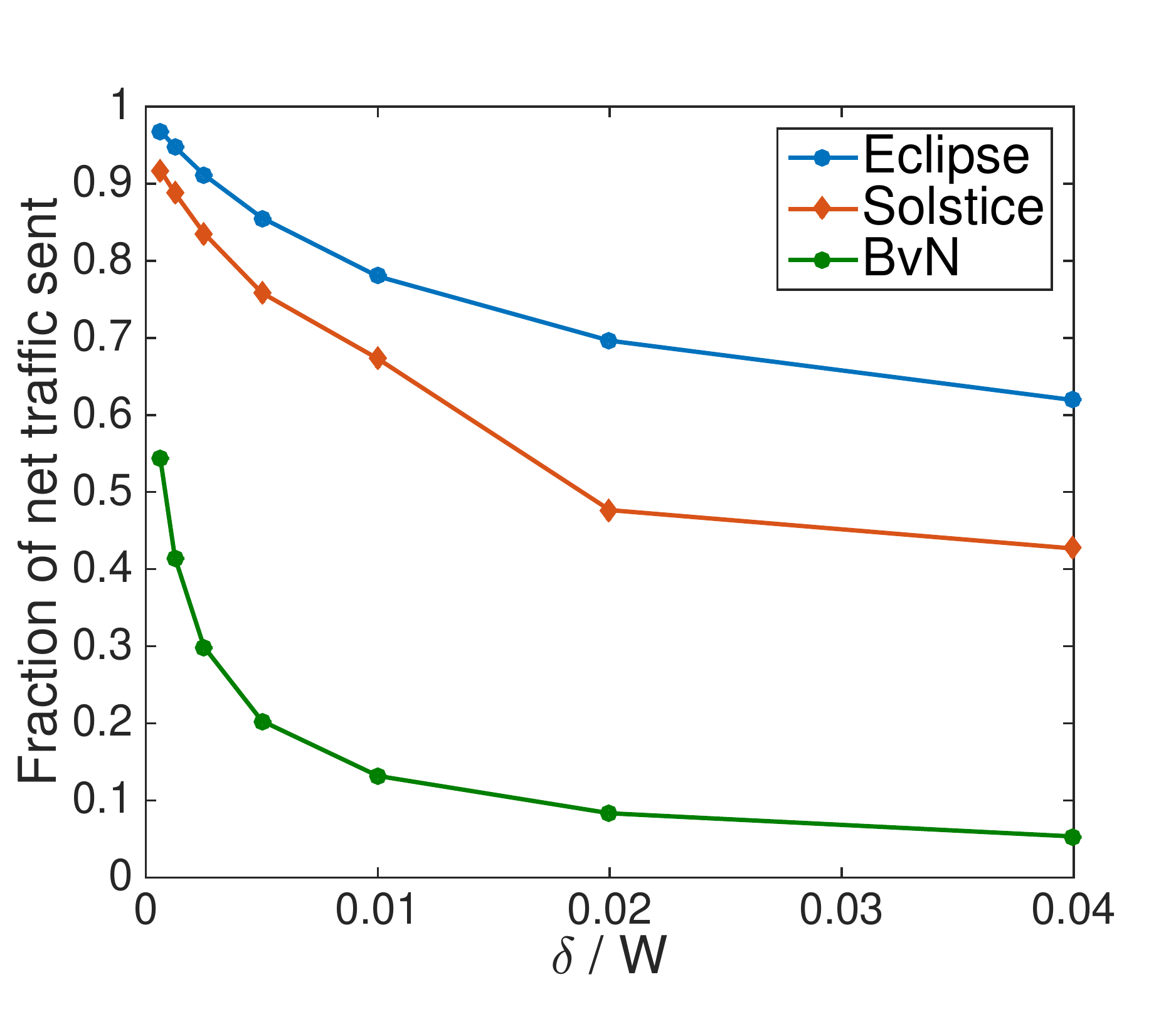}
     \label{fig:plot5}}
\hfil
     \subfigure[]{\includegraphics[height=50mm]{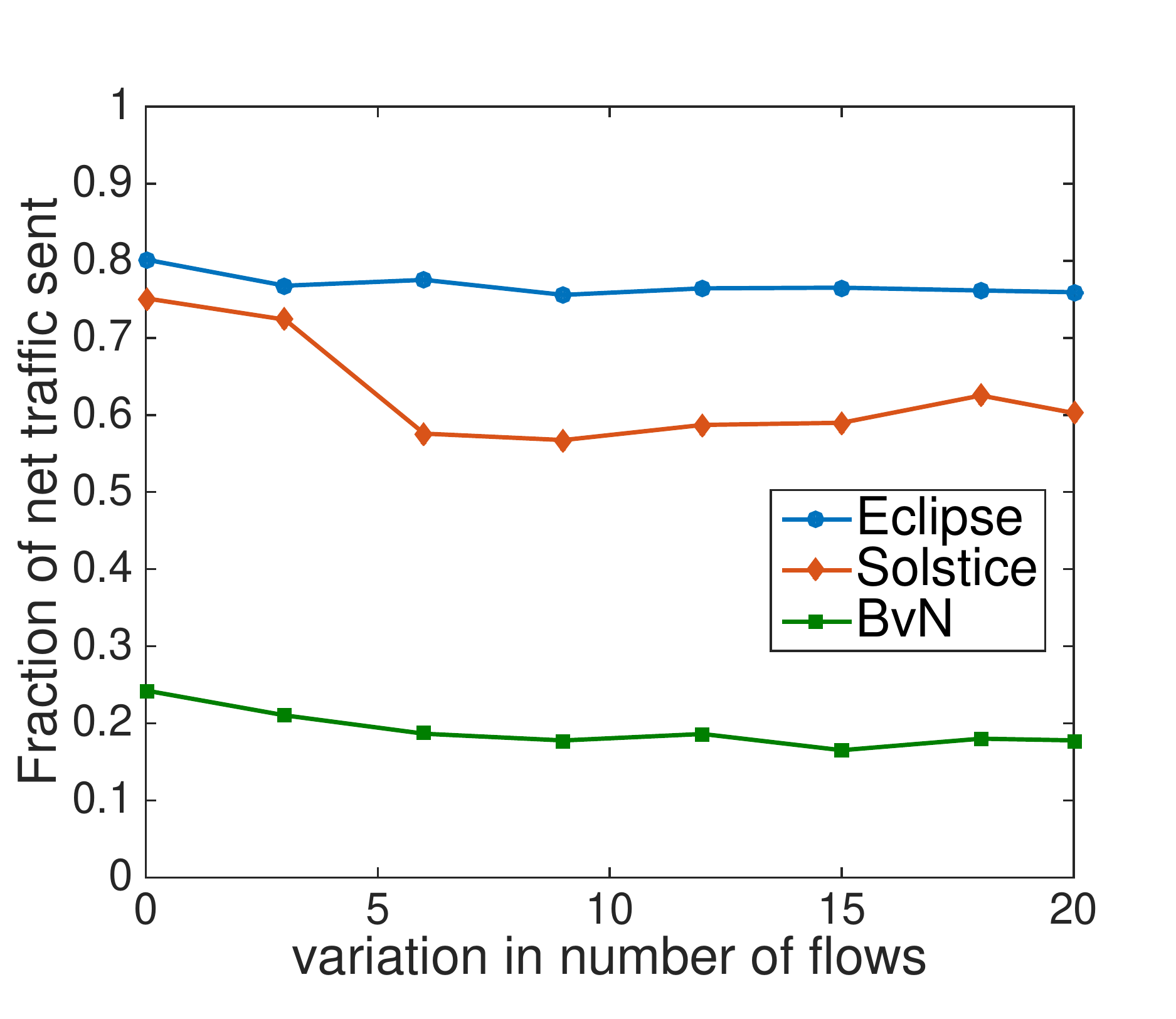}
     \label{fig:plot6}}
}
    \caption{Performance comparison of Eclipse under multi-block inputs.}
    \label{fig:plots4-6}
\vspace*{-2mm}
\end{figure*}

\begin{figure*}[t]
  \centerline{\subfigure[]{\includegraphics[height=50mm]{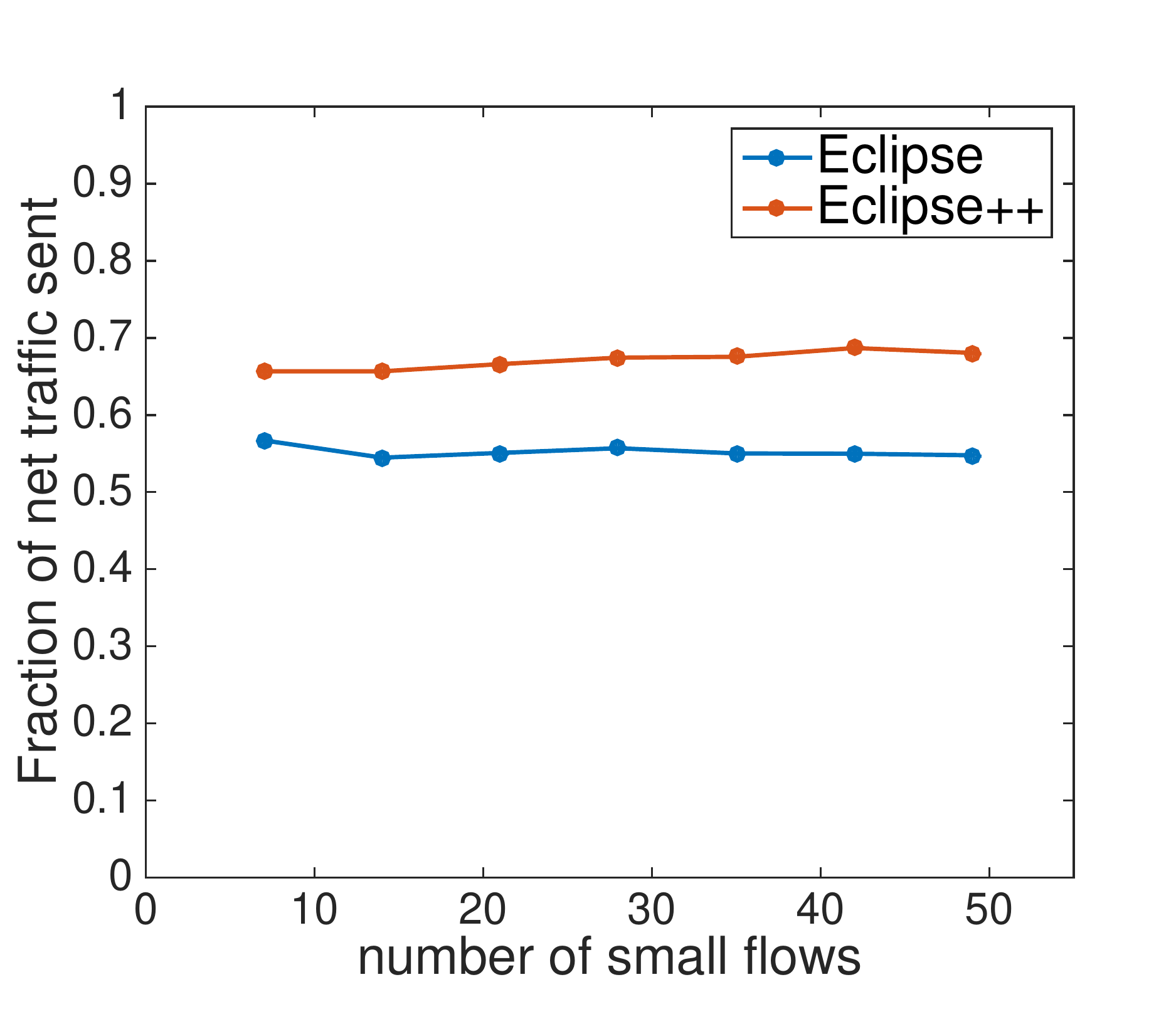}
     \label{fig:plot11}}
\hfil
     \subfigure[]{\includegraphics[height=50mm]{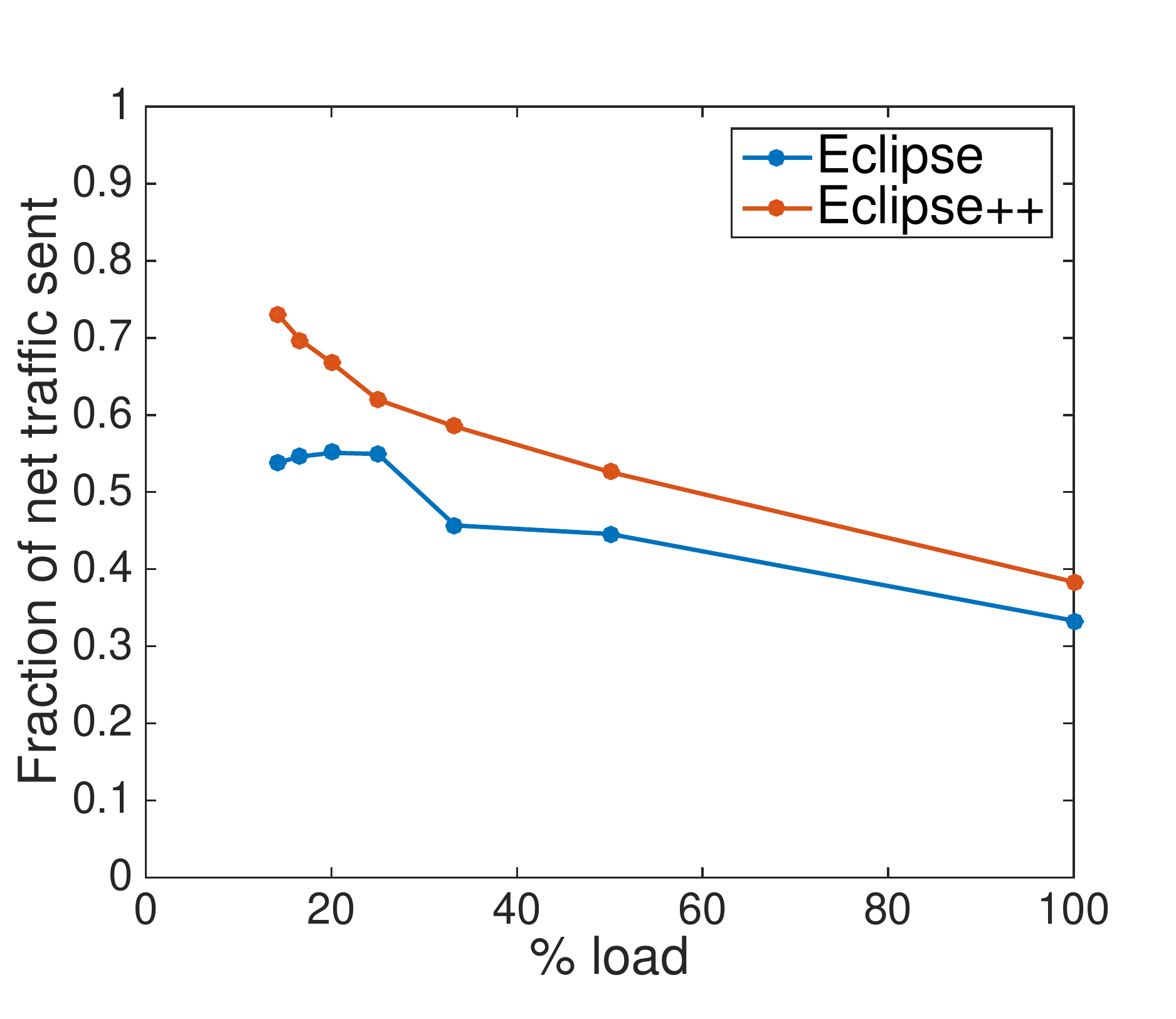}
     \label{fig:plot12}}
\hfil
     \subfigure[]{\includegraphics[height=50mm]{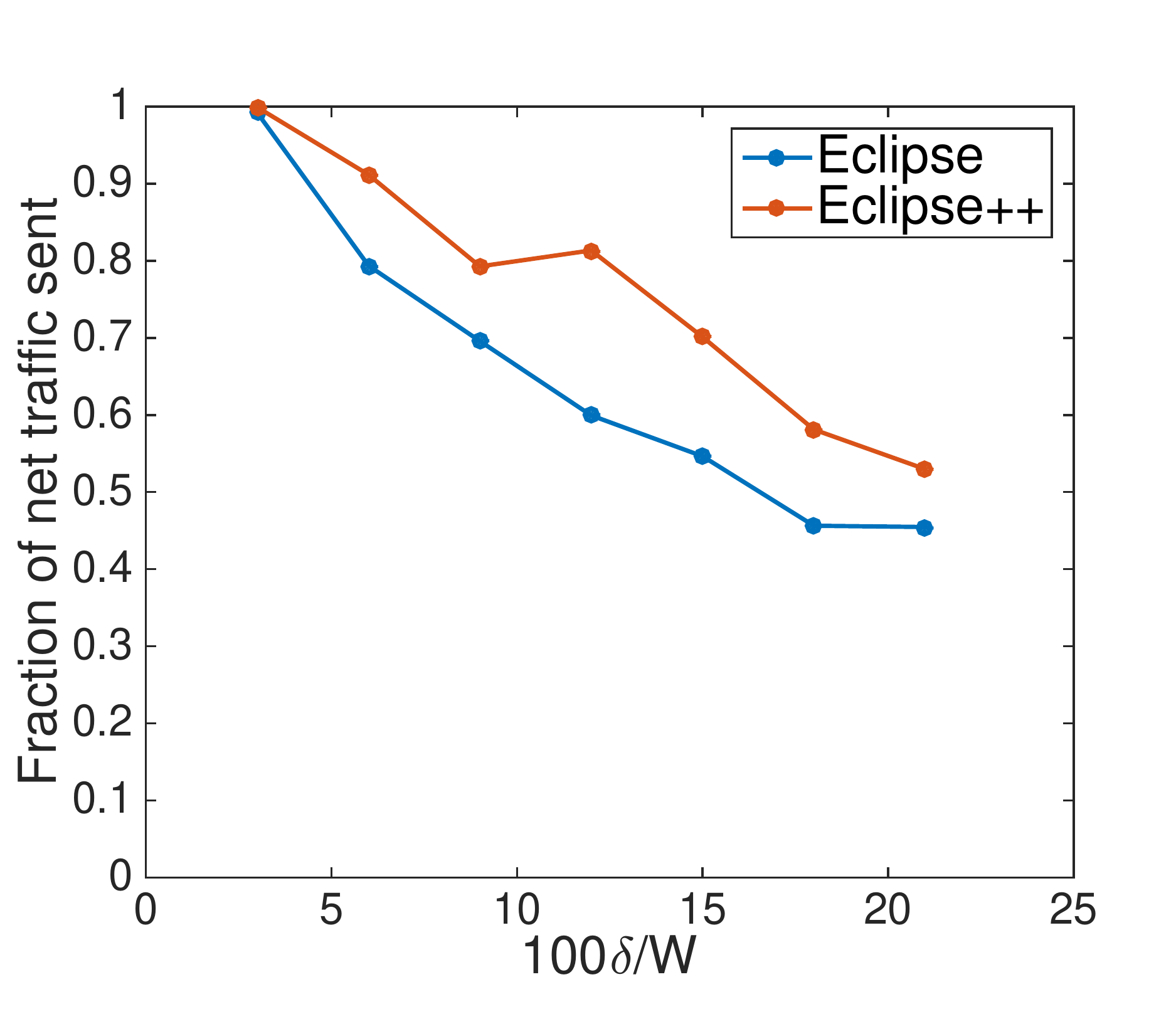}
     \label{fig:plot13}}
}
    \caption{Performance of Eclipse++ and Eclipse. Here Eclipse++ uses the schedule computed by Eclipse.}
    \label{fig:plots block 2}
\vspace*{-2mm}
\end{figure*}

\subsubsection*{Multi-Block Inputs}

We now consider a more complex traffic model for a 200 node network with block diagonal inputs of the form
\begin{align}
T =  \left[\begin{array}{ccc} B_1 & & \mathbf{0} \\ & \ddots & \\ \mathbf{0} & & B_m \end{array} \right]
\end{align}
where each of the component blocks $B_1,B_2,\ldots,B_m$ can have its
own sparsity (number of flows) and skew (fraction of traffic carried
by large versus small flows) parameters. The different blocks model
the traffic demands of different tenants in a shared data center
network such as a public cloud data center. To begin with, we consider
inputs with two blocks
$T = \left[ \begin{array}{cc} B_1 & \mathbf{0} \\ \mathbf{0} & B_2
  \end{array} \right]$
where $B_1$ is a $n_1\times n_1$ matrix with 4 large flows (carrying
70\% of the traffic) and 12 small flows (carrying 30\% of the traffic)
and $B_2$ is a $(200 - n_1)\times (200- n_1)$ matrix with uniform
entries (up to sampling noise).

\noindent
\textbf{Size of block:} Fig.~\ref{fig:plot4} plots the throughput as
the block size of $B_2$ is increased from 0 to 70. We observe a very
pronounced difference in the performance of Eclipse and Solstice:
Eclipse has roughly $1.5 - 2\times$ the performance of Solstice. These
findings are in tune with the intuition discussed in Section~\ref{sec:
  Motivation} --- the deteriorated performance of Solstice is due its
insistance on perfect matchings in each round.

\noindent
\textbf{Reconfiguration delay:} Fig.~\ref{fig:plot5} plots throughput
while varying the reconfiguration delay, for fixed size of $B_2$ to be
$50\times 50$. As expected, the throughput of Solstice and Eclipse
both degrade as the reconfiguration delay $\delta$ increases. However,
Eclipse throughput degrades at a much slower rate than Solstice. The
gap between the two is particularly pronounced for
$\delta/W \geq 0.02$, a numerical value that is well within range of
practical system settings.

\noindent
\textbf{Varying number of flows:}
In the final experiment we consider block diagonal inputs with 8
blocks of size $25 \times 25$ each. Each block carries
$10 + \lfloor \sigma*(U - 0.5)) \rfloor $ equi-valued flows where
$U\sim \text{unif }(0,1)$ and $\sigma$ is a parameter that controls
the variation in the number of flows. When increasing $\sigma$ from 0
to 20 we see from Fig.\ref{fig:plot6} that Eclipse is more or less
able to sustain its throughput at close to 80\%; whereas Solstice is
significantly affected by the variation.

\subsection{Indirect Routing}

In this section, we consider a 50 node network with traffic matrices
having varying number of large and small flows as before. We compare
the performance of the direct routing algorithm and the indirect
routing algorithm that is run on the schedule computed by Eclipse. To
understand the benefits of indirect routing, we focus on the regime
where the reconfiguration delay $\delta/W$ is relatively large and the
scheduling window $W$ is relatively long compared to the traffic
demand. This regime corresponds to realistic scenarios where the
circuit switch is not fully utilized,\footnote{Real data center
  networks often
  have low to moderate utilization (e.g,
  10--50\%)~\cite{benson2010network}.} but the reconfiguration delay
is large. In this setting of relatively large $\delta /W$, switch
schedules are forced to have only a small number of matchings, and
indirect routing is critical to support (non-sparse) demand matrices.
The following experiments numerically demonstrate the added gains of
indirect routing.

\noindent
\textbf{Sparsity:} Fig.~\ref{fig:plot11} considers a demand with 5
large flows and number of small flows varying from 7 to 49. The large
and the small flows each carry 50\% of the traffic. We let
$\delta = 16W/100$ and consider a load of 20\% (i.e., $W=5$, and
traffic load at each port is 1). We observe that the performance of
the Eclipse++ is roughly 10\% better than Eclipse.

\noindent
\textbf{Load:} As the network load increases (Fig.~\ref{fig:plot12}),
we see that indirect routing becomes less effective. This is because
at high load, the circuits do not have much spare capacity to support
indirect traffic. However, at low to moderate levels of load, indirect
routing provides a notable throughput gain over direct routing. For
example, we see close to 20\% improvement with Eclipse++ over Eclipse
at 15\% load.

\noindent
\textbf{Reconfiguration delay:} Finally, we observe the effect of
$\delta /W$ on throughput by varying $\delta$ from $3W/100$ to
$21W/100$. At smaller values of reconfiguration delay $\delta$ both
Eclipse and Eclipse++ are able to achieve near 100\% throughput. With
increasing $\delta$ both algorithms degrade with Eclipse++ providing
an additional gain of roughly 20\% over Eclipse.

\section{Final Remarks} \label{sec:conclusion}

We have studied scheduling in hybrid switch architectures with
reconfiguration delays in the circuit switch, by taking a fundamental
and first-principles approach. The connections to submodular
optimization theory allows us to design simple and fast scheduling
algorithms and show that they are near optimal --- these results hold
in the direct routing scenario and indirect routing provided switch
configurations are calculated separately. The problem of jointly
deciding the switch configurations {\em and} indirect routing policies
remains open. While submodular function optimization with nonlinear
constraints is in general intractable, the specific constraints in
Equations~\eqref{eq: submod
  const 1}, \eqref{eq: submod const 2} and~\eqref{eq: submod const 3}
perhaps have enough structure that they can be handled in a principled
way; this is a direction for future work.

In between the scheduling windows of $W$ time units, traffic builds up
at the ToR ports. This dynamic traffic buildup is known locally to
each of the ToR ports and perhaps this local knowledge can be used to
pick appropriate indirect routing policies in a distributed, dynamic
fashion. Such a study of indirect routing policies is initiated in a
recent work~\cite{cao2014joint}, but this work omitted the switching
reconfiguration delays. A joint study of distributed dynamic traffic
scheduling in conjunction with a static schedule of switch
configurations that account for reconfiguration delays is an
interesting direction of future research.

%\end{document}  % This is where a 'short' article might terminate

%ACKNOWLEDGMENTS are optional
\section{Acknowledgments}
The authors would like to thank Prof. Chandra Chekuri and Prof. R. Srikant for the many helpful discussions. 

\bibliographystyle{abbrv}
{\scriptsize
\bibliography{sigproc}}  % sigproc.bib is the name of the Bibliography in this case
% ACM needs 'a single self-contained file'!
%\balancecolumns

\appendix

\section{Direct Routing - Proofs}
In the following section we give the proof of Theorem~\ref{thm: single hop submod}. 

\subsection{Proof of Theorem~\ref{thm: single hop submod}}
\begin{proof}
Clearly $f(\{\}) = 0$. Also for any $S \subseteq S' \in 2^{\mathbb{R}_+\times \mathcal{M}}$  we have
\begin{align*}
\min \left\{ \sum_{(\alpha,M) \in S}\alpha M, T\right\} & \leq \min \left\{ \sum_{(\alpha,M) \in S'}\alpha M, T\right\} \\
\Rightarrow f(S) & \leq f(S').
\end{align*}
Hence $f$ is normalized and monotone. The following identity holds for reals $a,b,c$: $\min (a+b,c) = \min(a,c) + \min(b,c-\min(a,c))$. Therefore for $S\in 2^{\mathbb{R}_+\times \mathcal{M}}$ and $(\alpha_0,M_0)\notin S$, we have
\begin{align}
f(S \cup \{(\alpha,M)\}) =& \sum_{i,j\in[n]} \min \left\{ \sum_{(\alpha,M)\in S}\alpha M + \alpha_0 M_0, T \right\}_{i,j} \\
=& \sum_{i,j\in[n]}  \left\{ \min \left\{ \sum_{(\alpha,M)\in S}\alpha M, T \right\}_{i,j} + \right. \\
 \min\left\{ \alpha_0 M_0,  T \right. & \left. - \min \left. \left\{ \sum_{(\alpha,M) \in S}\alpha M, T \right\} \right\}_{i,j}\right\} \\
\Rightarrow f_S(\{(\alpha,M)\})  =& \left\| \min \left\{ \alpha_0 M_0,T - \min \left\{ \sum_{(\alpha,M) \in S}\alpha M, T \right\}\right\}\right\|_1 \label{eq: prop11}
\end{align}
Now, for $S \subseteq S' \in 2^{\mathbb{R}_+\times \mathcal{M}}$ and $(\alpha,M) \notin S'$ since
\begin{align}
T - \min \left\{ \sum_{i\in S'}\alpha_i M_i, T\right\} \leq T - \min \left\{ \sum_{i\in S}\alpha_i M_i, T\right\} \label{eq: prop1 2}
\end{align}
together with Equation~\eqref{eq: prop11} this implies
\begin{align}
f_{S'}(\{(\alpha_0,M_0)\}) \leq f_{S}(\{(\alpha_0,M_0)\}).
\end{align}
Hence $f$ is submodular.
\end{proof}
Next, we present the proof of Theorem~\ref{thm: submod apx guarantee}. 

\subsection{Proof of Theorem~\ref{thm: submod apx guarantee}}
\begin{proof}
Recall the submodular sum-throughput function $f$ defined in Equation~\eqref{eq:func defn}. Let $\{(\alpha_1,M_1),\ldots,(\alpha_k,M_k)\}$ be the schedule returned by Algorithm~\ref{alg: sha}. Let $S_i = \{(\alpha_1,M_1),$ $\ldots,(\alpha_i,M_i)\}$  denote the schedule computed at the end of $i$ iterations of the \texttt{while} loop and let $S^*$ denote the optimal schedule. Now, since in the $i+1$-th iteration $(\alpha_{i+1},M_{i+1})$ maximizes $\frac{\min(\alpha M, T_\mathrm{rem}(i+1))\|_1}{(\alpha + \delta)} = \frac{f_{S_i}(\{(\alpha, M)\})}{(\alpha + \delta)}$ we have for any $(\alpha, M) \notin S_i$, 
\begin{align}
\frac{f_{S_i}(\{(\alpha, M)\})}{(\alpha + \delta)} &\leq \frac{f_{S_i}(\{(\alpha_{i+1}, M_{i+1})\})}{(\alpha_{i+1} + \delta_{i+1})} \\
\Rightarrow f_{S_i}(\{(\alpha, M)\}) &\leq \frac{(\alpha + \delta)}{(\alpha_{i+1} + \delta_{i+1})} f_{S_i}(\{(\alpha_{i+1}, M_{i+1})\}). \label{eq: ana f bound}
\end{align}
Now consider $\mathtt{OPT} - f(S_i)$ for some $i<k$. Since $f$ is monotone we have
\begin{align}
\mathtt{OPT} - f(S_i) &= f(S^*) - f(S_i) \\
&\leq f(S_i \cup S^*) - f(S_i) \\
&\leq \sum_{(\alpha,M)\in J^*} f_{S_i}(\{(\alpha, M)\}) \\
&\leq \sum_{(\alpha,M)\in J^*} \frac{(\alpha + \delta)}{(\alpha_{i+1} + \delta_{i+1})} f_{S_i}(\{(\alpha_{i+1}, M_{i+1})\}) \label{eq: proof bound} \\
&\leq \frac{W}{(\alpha_{i+1} + \delta_{i+1})} f_{S_i}(\{(\alpha_{i+1}, M_{i+1})\}) \label{eq: proof bound 2}
\end{align}
where $J^* := S^* \backslash S_i $ denotes the set of matchings that are present in the optimal solution but not in $S_i$, Equation~\eqref{eq: proof bound} follows from Equation~\eqref{eq: ana f bound}, and Equation~\eqref{eq: proof bound 2} follows because $\sum_{(\alpha,M)\in J^*}(\alpha + \delta) \leq \sum_{(\alpha,M)\in S^*}(\alpha + \delta) \leq W$. Next, observe that 
\begin{align}
f(S_{i+1}) &= f(S_i) + f_{S_i}(\{(\alpha_{i+1},M_{i+1})\}) \\
\Rightarrow \mathtt{OPT} - f(S_{i+1}) &= \mathtt{OPT} - f(S_i) -  f_{S_i}(\{(\alpha_{i+1},M_{i+1})\}) \\
&\leq (\mathtt{OPT} - f(S_i))\left(1 - \frac{(\alpha_{i+1} + \delta)}{W}\right) \label{eq: proof an bound} \\
&\leq (\mathtt{OPT} - f(S_0)) \prod_{i'=1}^{i+1} \left(1 - \frac{(\alpha_{i'} + \delta)}{W}\right) \\
&\leq \mathtt{OPT} \times e^{-\sum_{i'=1}^{i+1} (\alpha_{i'} + \delta)/W} \label{eq: exp inequality}
\end{align}
where Equation~\eqref{eq: proof an bound} follows from Equation~\eqref{eq: proof bound 2} and Equation~\eqref{eq: exp inequality} follows because of the identity $1-x \leq e^{-x}$. Now, since after the $k$-th iteration the \texttt{while} loop terminates, this implies $\sum_{i'=1}^k (\alpha_{i'}+\delta) > W$. However, if the entries of the input traffic matrix $T$ are bounded by $\epsilon W + \delta$, then no matching has a duration longer than $\epsilon W$. In particular $\alpha_k + \delta \leq \epsilon W \Rightarrow \sum_{i'=1}^{k-1}(\alpha_{i'}+\delta) \geq W(1-\epsilon)$. Thus, setting $i = k-2$ in Equation~\eqref{eq: exp inequality} we have 
\begin{align}
\mathtt{OPT} - f(S_{k-1}) &\leq \mathtt{OPT} \times e^{-\sum_{i'=1}^{k-1} (\alpha_{i'} + \delta)/W} \\
&\leq \mathtt{OPT} \times e^{-(1-\epsilon)} \\
\Rightarrow \mathtt{OPT} - \mathtt{ALG2} &\leq  \mathtt{OPT} \times e^{-(1-\epsilon)}. 
\end{align} 
Hence we conclude $\mathtt{ALG2} \geq \mathtt{OPT}(1 - e^{-(1-\epsilon)})$. 
\end{proof}

\subsection{Correctness} \label{sec: Correctness}
Consider traffic matrix $T\in \mathbb{Z}^{n\times n}$. Let $\mathcal{T}=\{T(i,j): i,j\leq [n]\}$ denote the distinct entries in the matrix $T$. Then, in the following, we show that the maximizer in 
\begin{align}
\max_{\alpha \in\mathbb{Z}, M\in\mathcal{M}}\frac{\|\min(T,\alpha M)\|_1}{\alpha + \delta} 
\end{align}
occurs for $\alpha \in \mathcal{T}$. 

For any matching $M\in\mathcal{M}$ let us define $f_M(\alpha) \triangleq \|\min(\alpha M, T)\|_1$
and let 
$
f(\alpha) \triangleq \max_{M\in\mathcal{M}} \frac{f_M(\alpha)}{\alpha + \delta}
$
\begin{prop} \label{prop: correctness}
$f_M(\alpha)$ is (i) non-decreasing, (ii) piece-wise linear where the corner points are from $\mathcal{T}$ and (iii) concave. 
\end{prop}
\begin{proof}
It is easy to see (i) because if $\alpha_1\leq \alpha_2$ then $\min(\alpha_1 M,T)\leq \min(\alpha_2 M,T)$ entrywise and hence $f_M(\alpha_1) \leq f_M(\alpha_2)$. To see (ii) consider any $t_1 < t_2 \in \mathcal{T}$ such that no other element of $\mathcal{T}$ is between $t_1$ and $t_2$. Then for $t_1\leq \alpha \leq t_2$ we have 
\begin{align}
f_M(\alpha) = \|\min(\alpha M, T)\|_1 = \sum_{\substack{(i,j)\in M \\ T(i,j) \leq t_1}} \min(\alpha, T(i,j))  \\ 
+ \sum_{\substack{(i,j)\in M \\ T(i,j) \geq t_1}} \min(\alpha, T(i,j)) \\
= \sum_{\substack{(i,j)\in M \\ T(i,j) \leq t_1}} T(i,j) + \sum_{\substack{(i,j)\in M \\ T(i,j) \geq t_1}} \alpha \\
= \sum_{\substack{(i,j)\in M \\ T(i,j) \leq t_1}} T(i,j) + \left|\{(i,j)\in M: T(i,j) \geq t_1\}\right| \alpha \label{eq: correctness}
\end{align}  
Thus $f_M(\cdot)$ is linear for $t_1\leq \alpha \leq t_2$ and (ii) follows. (iii) also follows from Equation~\eqref{eq: correctness} by observing that 
\begin{align}
\left|\{(i,j)\in M: T(i,j) \geq t_1\}\right| \geq \left|\{(i,j)\in M: T(i,j) \geq t_2\}\right|
\end{align} 
for any $t_1 < t_2\in \mathcal{T}$. Hence the slope of the piece-wise linear function $f_M(\alpha)$ is non-increasing as $\alpha$ increases. In other words, $f_M(\alpha)$ is concave. 
\end{proof}
\begin{prop} \label{prop: correctness prop 2}
$\arg\max_{\alpha} \frac{f_M(\alpha)}{\alpha + \delta} \in \mathcal{T}$
\end{prop}
\begin{proof}
This follows from Proposition~\ref{prop: correctness}-(ii). Let $f_M(\alpha)$ be linear for $\alpha\in [t_1,t_2]$. Then it can be written as $f_M(\alpha) = f_M(t_1) + m(\alpha - t_1)$ for some slope $m\geq 0$. Now, consider the derivation of the function $f_M(\alpha)/(\alpha + \delta)$ in the interval $[t_1,t_2]$: 
\begin{align}
\frac{d}{d\alpha} \left(\frac{f_M(\alpha)}{\alpha + \delta} \right) = \frac{d}{d\alpha}\left(\frac{f_M(t_1) + m(\alpha - t_1)}{\alpha + \delta} \right) \\
= \frac{(\alpha + \delta)(m) - (f_M(t_1)+m(\alpha-t_1))}{(\alpha + \delta)^2} \\
= \frac{\delta m - f_M(t_1) + mt_1}{(\alpha + \delta)^2} \label{eq: correct slope}
\end{align}
Note that the numerator of Equation~\eqref{eq: correct slope} is independent of $\alpha$ and the denominator is strictly positive. Hence the sign (i.e., $>0, <0$ or $=0$) of the slope of $f_M(\alpha)/(\alpha + \delta)$ is the same in the interval $[t_1,t_2]$. This proves that the maxima must occur at either of the extreme points $t_1$ or $t_2$. By Proposition~\ref{prop: correctness}-(ii) we know that the $f_M(\alpha)$ is piece-wise linear with the corner points from the set $\mathcal{T}$. Thus we can conclude that the maxima must occur at one of the points in $\mathcal{T}$.  
\end{proof}
\begin{thm}
$\arg\max_{\alpha} f(\alpha) \in \mathcal{T}$ 
\end{thm}
\begin{proof}
This follows directly from Proposition~\ref{prop: correctness prop 2}. Notice that 
\begin{align}
\max_{\alpha} f(\alpha) = \max_\alpha \max_M \frac{f_M(\alpha)}{\alpha + \delta} \\
= \max_{M} \left( \max_\alpha \frac{f_M(\alpha)}{\alpha + \delta}\right)
\end{align}
But the maximizer of $f_M(\alpha)/(\alpha + \delta)$ belongs to $\mathcal{T}$ for any $M$. Hence we conclude that the maximizer of $f(\alpha)$ also belongs to $\mathcal{T}$. The Theorem follows. 
\end{proof}

\balancecolumns
% That's all folks!
\end{document}